\documentclass[11pt]{article}

\usepackage{amsmath, amssymb, amsthm}
\usepackage{graphicx}
\usepackage{xspace}
\usepackage{geometry}
\usepackage{booktabs} 
\usepackage[ruled,vlined]{algorithm2e} 
\usepackage{tabularx}
\usepackage{array}
\usepackage{mathtools}
\usepackage{makecell}
\usepackage{multirow}
\usepackage{nicefrac}
\usepackage{subcaption}
\usepackage{pgfplots}
\usepackage{dsfont}
\usepackage{tikz}
\usetikzlibrary{graphs,graphs.standard,calc,arrows.meta,bending}
\pgfplotsset{compat=1.17}
\usepackage{comment}
\usepackage{pifont}
\usepackage{wrapfig}
\usepackage{thm-restate}
\usepackage{bm}
\usepackage{natbib}
\usepackage[dvipsnames]{xcolor}
\usepackage[colorinlistoftodos]{todonotes}
\usepackage{diagbox}
\usepackage{enumitem}

\SetAlFnt{\small}
\SetAlCapFnt{\small}
\SetAlCapNameFnt{\small}
\SetAlCapHSkip{0pt}
\IncMargin{-\parindent}
\LinesNumbered
\SetAlgoSkip{smallskip}

\newtheorem{claim}{Claim}
\newtheorem{proposition}{Proposition}
\newtheorem{remark}{Remark}
\newtheorem{observation}{Observation}
\newtheorem{theorem}{Theorem}
\newtheorem{lemma}[theorem]{Lemma}
\newtheorem{corollary}[theorem]{Corollary}
\newtheorem{definition}{Definition}

\newtheorem{example}{Example}

\newcommand{\alloc}{{\mathcal{B}}}
\newcommand{\augment}{\cup}
\newcommand{\IBP}{item-based potential\xspace}
\newcommand{\IBR}{item-based realized\xspace}
\newcommand{\BBP}{bundle-based potential\xspace}
\newcommand{\BBR}{bundle-based realized\xspace}
\newcommand{\squiz}{{transfer$^+$ \xspace}}
\newcommand{\algoname}{{Bilevel Yankee Swap}}
\newcommand*\circled[1]{\tikz[baseline=(char.base)]{
            \node[shape=circle,draw,inner sep=1pt] (char) {#1};}}

\NewDocumentCommand{\cc}{ O{} O{} m }{\mbox{%
    \expandafter\ifx\expandafter\relax\detokenize{#2}\relax\else{#2-}\fi%
    \textsf{#3}%
    \expandafter\ifx\expandafter\relax\detokenize{#1}\relax\else{-#1}\fi%
    }\xspace}

\newcommand{\NPh}{\cc[hard]{NP}}

\newif\iflong
\newif\ifshort

\iflong
\else
\shorttrue
\fi

\title{Delegated Fair Division}

\author{Argyrios Deligkas\textsuperscript{1}, Michail Fasoulakis\textsuperscript{1}, Stavros D. Ioannidis\textsuperscript{1},\\
Alviona Mancho\textsuperscript{2,3}, and 
Evangelos Markakis\textsuperscript{2,3,4}
 \\\\
\textsuperscript{1} Royal Holloway University of London, United Kingdom\\
\textsuperscript{2} Athens University of Economics and Business, Greece \\
\textsuperscript{3} Archimedes, Athena Research Center, Greece \\
\textsuperscript{4} Input Output Group (IOG), Greece
}
\date{ }

\definecolor{AegeanBlue}{RGB}{1, 59, 150} 
\usepackage[
  colorlinks=true,
  citecolor=AegeanBlue,
  linkcolor=teal,
  urlcolor=teal
]{hyperref}

\begin{document}
\maketitle

\pagenumbering{gobble}
\begin{abstract}
Motivated by recently introduced problems on delegated resource allocation, we study a model of fair division, where a set of indivisible goods is to be allocated to some agents, each of which belonging to some bigger central entity. Our model captures the general framework of allocating resources to organizational units, which subsequently distribute them to their affiliated members. A particularly relevant application of this framework, with immense social impact, arises in the allocation of food donations through charitable organizations. In essence, every center acts as the representative of the agents belonging to it, aligning their generally different preferences. Our goal is to distribute the goods in a way that is simultaneously fair both with respect to the centers and the agents. We distinguish two different information structures depending on whether the agents compare their bundles against every agent or only those belonging to the same center. For each one of them we provide efficient algorithms that produce allocations that satisfy envy-based fairness guarantees at both levels.
\end{abstract}
\newpage

\pagenumbering{arabic}

\section{Introduction}
Fair division of indivisible resources has experienced significant growth during the last decades, as can be easily seen by recent surveys such as \cite{DBLP:journals/ai/AmanatidisABFLMVW23}.
Supporting applications have also contributed to the further development of the field, including course allocation algorithms \cite{budish2011combinatorial}, food donation programs \cite{Mertzanidis0V24}, and many others. Under indivisible resources, the main research direction that has pushed the field forward is to study solution concepts and algorithms for relaxations of envy-freeness. This includes, among others, the standard notions of EF1~\cite{budish2011combinatorial,LiptonMMS04} and EFX~\cite{caragiannis2019unreasonable,gourves2014near} (envy-freeness up to one or up to any item respectively), which is also the focus of our work.

In the vast majority of the literature, the allocation problem is presented as a question of assigning a set of resources to a set of independent agents. This is certainly adequate to model several instantiations of fair division problems. There are however many other settings where the allocation process is not performed directly on the agents, but instead the agents are represented via delegates who act on their behalf. This induces a two-layer allocation problem, where the delegates are in the upper level and act as distribution centers for the agents in the lower level. To mention some  natural motivating applications, food donation or other social purpose programs fit well under this framework. Charity organizations like Feeding America, provide resources and meals to partner food banks, which in turn allocate them locally to individuals registered with them. Likewise, university departments may receive resources from the university itself, which are then allocated to research teams. In yet another example, the European Union distributes resources at country-level, which are then allocated to domestic organizations or ministries by each country. 

The previous discussion introduces a new dimension in fair allocation problems which comes with its own challenges. First of all, the upper-level agents, that we will refer to as {\it centers} from now on, are agents themselves with their own preferences. E.g., the food banks in the food donation programs may not consume any meals but still derive a value dependent on the set of resources they receive from charities. 
In a similar manner, the departments of a university derive a value when there is a productive utilization of the resources by its faculty members.
What makes the problem even more interesting is the potential interdependency between the levels, since the valuation of the centers could depend not just on the resources they distribute, but also on the valuations of the agents they represent. 
E.g. a university department can be happier as the utility of its faculty members increases. A further consequence of having these two levels of entities is that fairness notions now should be defined both at the level of centers but {\em also} for the agents at the lower level. And especially for the latter, the type of fairness that would be more suitable depends on the application. Namely, one could demand fairness across {\em all} agents or {\em only} among the agents within a center, since agents underneath a center may not even observe what other centers possess, due to informational and privacy constraints; the agents in a local food bank will not even know what resources are allocated to a different food bank.  

Obviously, the above give rise to several avenues for further exploration, with the eventual goal of providing satisfactory fairness guarantees at both levels. To our knowledge, apart from the notable exception of \cite{LBBM25} (discussed further in our related work), such models have not been studied in other works on fair division before.

\subsection{Contribution}
We view as our conceptual contribution the proposal of a  two-level allocation problem and the study of algorithms that guarantee envy-based fairness notions, such as appropriate adaptations of EF1 and EFX. 
The model we introduce allows for a wide range of flexibility that can capture various scenarios on how to define envy and fairness across both levels. As a starting point, in Section \ref{sec:prelims}, we distinguish four different variations that concern the definition of the centers' valuation function, which in turn affects how the centers perceive envy. These variants have to do with the way that a center evaluates the bundling of another center (i.e., on whether it can reshuffle the goods of the other center, referred to as {\it item-based}, or think of them as fixed bundles that cannot be altered, referred to as {\it bundle-based}) and also with how a center evaluates its own bundling.

Moving on to our technical contribution, in Section \ref{sec:intra}, we start with some warm-up results on  {\it intra-fairness} criteria. By this, we refer to enforcing fairness among the agents within each center, while also having fairness guarantees among the centers. 
This is an easier case, where we are able to show that at least for one of our model variants, we can guarantee existence for any fairness notion for which existence is known in the classic (one level) fair division model.

We then move to our main technical results in Sections \ref{sec:inter} and~\ref{sec:yankee} that concern {\em inter-fairness}, where agents compare their bundle against every other agent. 
We provide an impossibility result, ruling out EFX among the centers in conjunction with EF1 among the agents. To complement this, we prove that deciding whether an instance of our setting admits such an allocation is \NPh.

Therefore, our main focus in these sections is on obtaining EF1 both among the centers but also among all the agents. This is still a very challenging problem, in contrast to the classic setting where EF1 for a single level of agents is easily achieved \cite{LiptonMMS04}. Consequently, we provide positive results only for special classes of valuation functions. In particular, we are able to obtain an EF1 allocation across both levels for most of our variants when the agents agree on the ranking of the goods w.r.t. their value within each center.
 This is a well motivated family of valuations that has been considered in the past literature, and our result is based on an appropriate modification of the Round Robin algorithm. 
{Furthermore, under the stronger assumption that all agents are identical, we show that EFX can be achieved at the agent level in conjunction with EF1 at the center level.}

We then move to bivalued valuations. Despite their simplicity, this is a trickier case,  highlighting already the main difficulties with existing algorithms, 
and where our results differ across our variants. 
For \IBR valuations, we show that EF1 may not even be attainable at both levels, as opposed to the case of \IBP valuations.
For the bundle-based scenario, we view as our most technical contribution, that we make progress for one of our variants on the relaxation to binary preferences, 
establishing EF1 on both levels. This result relies on path augmentation techniques, inspired by the Yankee Swap algorithm \cite{viswanathan2025general}. Our insight here is that although the original Yankee Swap algorithm cannot work for our setting, an appropriare adaptation that also includes the centers in the transfer paths is suitable.

Finally, as an ending note, and in order to further demonstrate the difficulty of obtaining guarantees for some simple special cases of our problem, we show that such results can yield improvements for certain versions of the classic fair division setting without levels. In Section \ref{sec:conclusions}, Paragraph~\ref{open:one-to_one}, we expand further on this. 

\subsection{Related work}
The conceptually most related work to ours is that of
\cite{LBBM25}. This is a very recent independent work, completed in parallel to ours and  concerns multi-level allocation problems. The main difference is that it focuses on welfare notions, such as the Nash welfare and does not consider any envy-based fairness notions. We therefore consider our results to be orthogonal to \citep{LBBM25}.

Regarding technical similarities with related works, particularly our results in Section~\ref{sec:yankee}, we note that 
many algorithms
based on  path augmentation techniques exist in the literature~\citep{schrijver2003combinatorial}. In fair division, a notable example that inspired our work is the Yankee Swap algorithm, which we had to adapt in our model. For the classic setting of fair division without levels, the Yankee Swap algorithm was proved to establish EFX allocations for agents with binary submodular valuations among other properties~\citep{viswanathan2025general}. Interestingly, the work of \cite{LBBM25} has also made use of similar ideas.

Other relevant models in fair division involve {\it group fairness}. This is motivated by scenarios where resources are allocated to \emph{groups}, such as families or institutions, whose members jointly consume the same bundle of goods but may value the resources differently. Consequently, a line of work has evolved on the \emph{(fixed-)group} model, in which agents are partitioned into groups and the goal is to allocate resources fairly with respect to the groups \citep{DBLP:conf/ijcai/BenabbouCEZ19,manurangsi2022almost,manurangsi2025ordinal,kyropoulou2020almost,DBLP:journals/corr/abs-2502-10516,golz2025fair,manurangsi2026tight}. As an example, \cite{manurangsi2022almost} studied envy-freeness up to $c$ goods (EF$c$), which requires that no agent prefers another group’s bundle to her own group’s bundle after the removal of at most $c$ goods from the other bundle. As another variation, it is also conceivable that group membership is not fixed in advance. This motivated \cite{kyropoulou2020almost} to propose the \emph{variable-group} model, in which the partition of agents into groups is chosen alongside the allocation. Regarding the relation between the (fixed-)group model and our setting, one may think of centers as fixed groups of agents. However, the fundamental difference in our model, is that the resources are treated as private goods: an agent does not derive value from goods allocated to other agents and they do not jointly consume goods.

Another relevant line of work involves \emph{externalities}, where an agent’s value may depend on the bundles allocated to other agents~\cite{JM96}. Within fair division, there has also been a more recent line of works on allocations  of indivisible goods with externalities \citep{10.1145/3308558.3313670,mishra2022fair,aziz2023fairness,DBLP:conf/aaai/DeligkasEKS24,connor2026tight}. 
In particular, \cite{aziz2023fairness} studied EF1 in the presence of externalities and established existence results for several restricted settings. Very recently, \cite{connor2026tight} established that EF1 allocations need not exist even when agents have binary preferences and that the strongest achievable relaxation is EF$k$ for $k = \Theta(\sqrt{n})$, where $n$ is the number of agents. 
In our setting, one can view the centers as ``special agents'' whose valuations depend on the welfare of their members, thereby exhibiting a form of externality, while the agents under the centers do not exhibit such behavior. None of the above results however have any implications for our model.

We also discuss some works on the classic fair division setting without levels, that are most relevant to the fairness notions that we study. 
One of the first relaxations of envy-freeness that was studied is the EF1 notion, formally introduced by \cite{budish2011combinatorial}, although also implicitly considered earlier by \cite{LiptonMMS04}. An EF1 allocation is always guaranteed to exist and can be computed in polynomial time \citep{LiptonMMS04}. The stronger fairness notion of EFX, was later proposed by \cite{gourves2014near,DBLP:journals/teco/CaragiannisKMPS19}, which has since attracted significant attention.
Positive results for EFX were first established by \cite{plaut2020almost}, showing existence when all agents share identical, not necessarily additive, valuation functions. They also showed that EFX allocations can be computed efficiently when agents agree on the ranking of items by value. Existence results for three agents with additive valuations were established by \cite{ChaudhuryGM24}, with a simpler proof derived later  by \cite{AkramiACGMM25}. Additional existence results have been proven for bivalued instances \citep{amanatidis2021maximum}.  Very recently, the general existence question for EFX was resolved negatively beyond additive valuations. \cite{akrami2026counterexampleefxnge} showed that EFX allocations need not exist for monotone valuations, already for three agents and eight goods; their construction extends to instances with $n \geq 3$ agents and $m \geq n+5$ goods and, also implies nonexistence for submodular valuations. Subsequently, \cite{mackenzie2026counterexamplesefxsubmodularsubadditive} gave a compact counterexample for submodular valuations and proved inapproximability results for subadditive valuations.

Due to the difficulties of achieving exact EFX, research has also focused on approximation algorithms. The current best-known approximation is $\phi - 1 \approx 0.618$, due to \cite{AMN20}. Moreover, there have been further improvements for special cases. For instance, \cite{markakis23improved} proposed a $2/3$-EFX algorithm for cases where agents agree on the top $n$ goods, with $n$ being the number of the agents. Additional scenarios achieving a $2/3$ approximation are discussed in \citep{AFS24}. General frameworks for constructing approximate EFX algorithms are also discussed in \citep{markakis23improved} and \citep{farhadi2021almost}.

Apart from EF1 and EFX, other fairness notions have also been studied in the literature, including Maximin Share (MMS) and its variants PMMS and GMMS, introduced by \cite{budish2011combinatorial}, \cite{DBLP:journals/teco/CaragiannisKMPS19}, and \cite{barman2018groupwise}, respectively. We refer the reader to \citep{DBLP:journals/ai/AmanatidisABFLMVW23} for a comprehensive overview of additional fairness notions and open problems.

\section{Model and preliminaries}
\label{sec:prelims}
We consider a scenario where we have $k$ entities, referred to as the {\it center} agents or simply centers, denoted by $C_1,\dots, C_k$. Each center can be seen as an entity that intends to allocate items to agents that are affiliated with the center. Namely, for every center we assume that there are $n$ affiliated agents (unless otherwise specified, we will stick to all centers having the same number of affiliated members). We will often denote the agents affiliated with center $C_j$ as $a_{(1,j)}, a_{(2, j)}, \dots, a_{(n, j)}$ and use $N_j$ to refer to the set of agents belonging to center $C_j$. 

Along with the agents, there is also a set of indivisible goods $M = [m]$\footnote{We use the notation $[z]$ to denote the set $\{1,2,\ldots ,z\}$, for $z\in \mathbb{N}$.} to be allocated. An allocation here is a 2-layered partition, that determines the items allocated to each center as well as the bundles given to the agents of each center. This is described below.
\begin{definition}\label{def:allocation}
An \emph{allocation}, or \emph{bundling profile}, of $M$ is 
a tuple $\mathcal{B} = (B_1,\dots,B_k)$, where: 
\begin{enumerate}
    
\item[(a)]
Each $B_j$ is a \emph{bundling}, i.e., a vector of disjoint sets assigned to the agents affiliated with center $C_j$,  namely, $B_j = \bigl(A_{(1,j)}, A_{(2,j)}, \dots, A_{(n, j)}\bigr)$,
    with $A_{(\ell,j)}$ being the \emph{bundle} given to agent $a_{(\ell,j)}$; 
\item[(b)]
    All items are allocated, i.e., 
    $\bigcup_{j\in [k]} \bigcup_{\ell\in [n]} A_{(\ell,j)} = M$.
\end{enumerate}
\end{definition}

The next step in setting up our model is to define how the centers and the agents underneath them value an allocation and how 
centers evaluate the bundles of other centers. 
For the agents that are underneath a center, we can use any of the standard valuation functions, as used in the fair division literature (e.g., monotone, additive, etc). As a starting point in this work, we will focus on additive valuations of the agents, unless specified otherwise. This means that for an agent $a$, receiving a set $S$ of goods, her value equals $\sum_{g\in S} v_a(\{g\})$. We sometimes refer to monotone valuations of the agents, meaning that for all $S, S' \subseteq M$, if $S \subseteq S'$, then $v_a(S) \leq v_a(S')$. We will also use $v_a(g)$ instead of $v_a(\{g\})$ for simplicity.

Given an allocation $\mathcal{B}$, we need to define both the value that a center has for the goods allocated to itself, and the value that it has for the items allocated to another center. These steps are necessary so that we can define what it means for a center to experience envy towards another center and to define envy-based fairness notions among centers.

There are several possibilities on how to proceed and we start here with some natural first approaches. Given that the centers are entities that would generally care for the best utilization of the resources they give to the agents, we focus on the utilitarian social welfare as the metric of economic efficiency that the centers are interested in\footnote{We leave as future work the study of other economic efficiency metrics as well (e.g. Nash welfare etc).}.

\medskip \noindent
{\bf Evaluating a bundling of a different center.} 
Let us first see how would a center $C_i$ evaluate the goods allocated to a different center, say $C_j$. Suppose that center $C_j$ has been given the bundling $B_j = (A_{(1,j)}, A_{(2,j)}, \dots, A_{(n, j)})$, and these sets are allocated to its $n$ agents. When center $C_i$ sees the allocation to $C_j$, it makes a thought experiment of how would these items produce maximum social welfare if it were to allocate them to its own agents. 
There are essentially two choices regarding this thought experiment of center $C_i$.
\begin{itemize}
    \item {\bf{\em Bundle-based potential value}.} The first choice is to allocate the {\em existing} bundles of center $C_j$, i.e., $(A_{(1,j)}, A_{(2,j)}, \dots, A_{(n, j)})$ to its own agents without altering them. Formally, let $\Pi_n$ be the set of all permutations on $[n]$. If $v_{a_{(\ell, i)}}(\cdot)$ is the valuation function of agent $a_{(\ell, i)}$, then 
    $v_{C_i}(B_j) = \max_{\pi \in \Pi_n} \sum_{\ell=1}^{n} v_{a_{(\ell, i)}}(A_{(\pi(\ell), j)})$.
    In other words, the center computes a matching of the existing bundles to its agents, in order to maximize their welfare.
    \item {\bf{\em Item-based potential value.}} The second choice is to completely redistribute the items $\cup_{\ell \in [n]}A_{(\ell,j)}$ into $n$ {\em new} bundles to its own agents. Hence, $v_{C_i}(B_j)$ will equal the maximum social welfare than can be achieved by such a redistribution.
\end{itemize}

\noindent
{\bf Evaluating a center's own bundling: realized value versus potential value.}
It is important to also distinguish the possible ways that a center evaluates the bundling allocated to its own agents. In particular, suppose that an algorithm has allocated the bundling $B_i = (A_{(1,i)}, \dots, A_{(n, i)})$ to the agents underneath center $C_i$. 
Here depending on the application scenario, there are two ways in which the center could evaluate $B_i$. The first one is to be consistent with the thought experiment and the models outlined above and use the potential value. 
I.e., if we adopt the bundle-based version, then center $C_i$ would evaluate $B_i$ by finding the welfare-optimal allocation of the bundles to its agents. This does not mean that the actual allocation to the agents will change. It is simply that the center perceives as its value the best possible welfare\footnote{To further support this, we can think of examples where a center evaluates the items it has received before an allocation to the agents is realized, e.g., a department could evaluate its total set of resources acquired (servers, lab equipment, etc) before an actual allocation to the research teams is realized.} that could be achieved if the agents were to exchange their bundles. 
Alternatively, it is also meaningful to consider a second choice, which is that the center simply calculates the actual welfare produced by the bundling $B_i$.
This gives rise to the following choices. 
\begin{itemize}
\item {\bf{\em Potential value for own bundling}}. In this case, $v_{C_i}(B_i)$ is defined in the same way as the bundlings of other centers, explained above. 
\item {\bf {\em Realized value for own bundling}}. Here, the value of center $C_i$ is 
$v_{C_i}(B_i) = \sum_{\ell=1}^{n} v_{a_{(\ell, i)}}(A_{(\ell, i)}).$
\end{itemize}

Since, as described earlier, every center always evaluates another center's bundling according to its potential value (whether bundle-based or item-based), the only remaining choice in order to characterize the valuation function of a center is how this center evaluates its own bundling: either by its potential value or its realized value. Combined with the bundle-based/item-based distinction, this yields four valuation functions: \BBP, \BBR, \IBP, and \IBR.
It is also easy to see that any guarantees we get for the realized value continues to hold for the potential value as well.

At this point, it is useful to formalize the notion of monotonicity for a center's valuation function.
To this end, similarly to the notion of a superset for sets, we first define \emph{superbundlings}.

\begin{definition}\label{superbundling}
    Let $B=(A_1, A_2, \ldots, A_n)$ be a bundling comprising bundles $\{A_i\}_{i\in[n]}$. We say that a bundling $B'=(A'_1, A'_2, \ldots, A'_n)$ is a \emph{superbundling} of $B$ if $A_i \subseteq A'_i$ for all $i\in[n]$.
\end{definition}

\begin{definition}\label{def:center-monotonicity}
    A center's valuation function $v_{C_j}$ is \emph{monotone} if, for any bundling $B$ and any superbundling $B'$ of it, we have $v_{C_j}(B) \leq v_{C_j}(B')$.
\end{definition}

\subsection{Envy-based fairness concepts}

Having settled the definition of valuations for both agents and centers, we now turn to notions of envy-based fairness. Since our goal is to compute allocations that are {\em simultaneously} fair with respect to both the centers and the agents, we need to define fairness notions for {\em both} entities. 
We begin with the agents. In our setting, the presence of centers introduces an additional layer, raising the  question of whether fairness among agents should be evaluated globally \emph{across} all agents of all centers, or locally \emph{within} each center. 
This gives rise to two corresponding approaches to fairness among agents, which we refer to as \emph{inter-fairness} and \emph{intra-fairness}, respectively. 

We begin with inter-fairness. Among the various notions that can be defined, we are particularly interested in \emph{inter-EF1} and \emph{inter-EFX}, that in essence correspond to the standard EF1 and EFX notions. 
According to the EF1 notion an agent may be envious of another agent's bundle, as long as there exists \emph{some} good in the other agent’s bundle whose hypothetical removal would eliminate this envy~\cite{budish2011combinatorial,LiptonMMS04}. 
The EFX criterion imposes a stronger requirement: envy must be eliminated with the hypothetical removal of \emph{any} good from the other agent’s bundle~\cite{DBLP:journals/teco/CaragiannisKMPS19}.

\begin{definition}[Agent inter-fairness]\label{def:inter-ef1-efx}
    Let $\mathcal{B} = (B_1, B_2, \ldots, B_n)$ be an allocation of the goods to the agents and let $A_{(i,j)}$ denote the bundle of agent $a_{(i,j)}$ that belongs to center $C_j$. Then, 
    $\mathcal{B}$ is:
    \begin{itemize}[leftmargin=*]
        \item {\em inter-EF1}, if for every pair $a_{(i, j)}, a_{(i',j')}$, 
        $\exists g\in A_{(i',j')}$, such that
$v_{a_{(i,j)}}(A_{(i,j)}) \geq v_{a_{(i,j)}}(A_{(i', j')} \setminus g)$;
        \item {\em inter-EFX}, if for every pair $a_{(i, j)}, a_{(i',j')}$ and $\forall g\in A_{(i',j')}$, we have $v_{a_{(i,j)}}(A_{(i,j)}) \geq v_{a_{(i,j)}}(A_{(i', j')} \setminus g)$.
    \end{itemize}
\end{definition}

In an analogous manner, we also define {\em intra-EF1} and {\em intra-EFX}, by adapting the above definitions so that they involve only agents within the {\em same} center. Obviously, intra-fairness is a  weaker notion than inter-fairness. 

Finally, it remains to address fairness from the perspective of the centers. 
In order to define EF1 and EFX for the centers, we propose a natural extension based on a thought experiment analogous to the case of agents. When a center compares its bundling to that of another center, envy is evaluated by hypothetically removing a \emph{single} good (for EF1) or \emph{any} good (for EFX) from one of the bundles of the other center’s bundling. In what follows, if the good $g$ being removed belongs to bundle $A_{(\ell,j)}$ of bundling $B_j$, we denote $B_j\setminus g = (A_{(1,j)},\ldots, A_{(\ell,j)}\setminus \{g\}, \ldots, A_{(n,j)})$.

\begin{definition}[Center fairness]\label{def:center-ef1-efx}
    Let $\mathcal{B}$ be an allocation, as described in Definition~\ref{def:inter-ef1-efx}. Then,  
    $\mathcal{B}$ is:
    \begin{itemize}[leftmargin=*]
        \item {\em EF1} for the centers, if for every pair  
        $C_i, C_j$,
        $\exists \ell \in [n]$, $\exists g\in A_{(\ell,j)}$, such that $v_{C_i}(B_i) \geq v_{C_i}(B_j \setminus g)$;
        \item {\em EFX} for the centers, if for every pair 
        $C_i, C_j$,
        $\forall \ell \in [n]$,
        $\forall g\in A_{(\ell,j)}$,
        we have $v_{C_i}(B_i) \geq v_{C_i}(B_j \setminus g)$.
    \end{itemize}
\end{definition}

When referring to a fairness notion $F_1$ at the center-level and simultaneously to a (possibly different) fairness notion $F_2$ at the agent-level, we will often write $\frac{F_1}{F_2}$. 
For example, $\frac{\text{EFX}}{\text{inter-EF1}}$ denotes EFX among centers and inter-EF1 among agents.
\section{Warm up: Intra-fairness} 
\label{sec:intra}

As a warm-up, we start with \emph{intra-fairness}, where agents compare their bundles only with those of other agents in the same center. 
As already discussed in the Introduction, requiring intra-fairness has a natural motivation, since it is very likely that in some cases the agents underneath a center may not even observe what the agents in other centers possess, due to informational and privacy constraints. Positive results on intra-fairness also serve as a benchmark on what we can achieve in this model, given that it is easier to satisfy than inter-fairness.

Our next result implies that at least for item-based valuations of the centers and for some generalizations, we can achieve for intra-fairness anything that we already know for the classic fair 
division setting without levels.

\begin{theorem}\label{thm:intra:item-wise}
    Assume that each center has a valuation function dependent only on the subset of items allocated to it \footnote{That is, the valuation of a center does not depend on the structure of the bundles or on how these bundles are allocated to its agents.} and belonging to some class $C$ of valuations. Assume also that the agents have valuations that belong to some class $C'$ of valuations.
    Let $F_1$ and $F_2$ be two fairness notions that are guaranteed to exist in the classic fair division setting, for the classes $C$ and $C'$ respectively. 
    Then there exists an $\frac{\text{F1}}{\text{intra-F2}}$ allocation. 
   \end{theorem}

\begin{proof}
    Consider two fairness notions $F_1$ and $F_2$, and assume that we want the final allocation to satisfy $F_1$ for the centers and intra-$F_2$ for the agents. Let $A$ and $B$ be the algorithms that compute allocations satisfying $F_1$ and $F_2$ respectively, in the standard fair division setting under monotone valuations of the agents. We construct the allocation as follows:
    
    \begin{itemize}
        \item Step 1. We determine the subset of goods each center receives. We run Algorithm $A$ treating each center $C_j$ for $j \in [k]$ as an agent, with $M$ as the set of items. This produces a partition of items $\mathcal{G} = \{G_1, \dots, G_k\}$, where each center $C_j$ receives the subset $G_j$. By the guarantee of Algorithm $A$, this partition satisfies $F_1$ for the centers.

        \item Step 2. We distribute the items within each center separately. For each center $C_j$, let $N_j$ be the set of agents that belong to it. We run Algorithm $B$ on the set of agents $N_j$ with $G_j$ as the set of available items. This constructs the bundles $\{A_{(i,j)}\}_{i\in [n]}$. By the guarantee of Algorithm $B$, this internal distribution satisfies  intra-$F_2$ for the agents.
    \end{itemize}

    Since the valuation of each center is independent of the way the items it received are allocated to its agents, the second step does not affect the value obtained by any center in the first step. Therefore, the final allocation simultaneously satisfies both fairness notions.
\end{proof}
 
\begin{remark}
In Theorem \ref{thm:intra:item-wise}, the requirement for equally-sized centers can be dropped, since the value of a center depends only on the set of items it receives and not on the number of its agents.
\end{remark}

\noindent Next, we show that the previous theorem applies to \IBP valuations of the centers. First, we establish the monotonicity property of each center’s valuation function, under monotone valuations of the agents. For completeness, we show that all the valuation variants introduced above are monotone in this case. We will use this result to argue that when the agents have monotone valuations, not only Theorem~\ref{thm:intra:item-wise} applies to \IBP valuations, but in fact we can guarantee any combination of an agent-level fairness notion and a center-level fairness notion, guaranteed to exist under monotone valuations in the standard setting. The proofs of the following two lemmas are deferred to Appendix~\ref{app:sec:intra}.

\begin{lemma}\label{lemma:intra:monotonicity}
     When the agents underneath the centers have monotone valuations, then the \IBP, \IBR, \BBP, \BBR valuations of the centers are monotone. 
\end{lemma}

The following lemma demonstrates that when agents have additive valuation functions, this additivity extends to the centers under \IBP valuations. This is particularly useful when adapting existing fair division results to our setting. 

\begin{lemma} \label{lemma:additivity-ibp}
    If the agents within a center have additive valuations, then the \IBP valuation of that center is additive. 
\end{lemma}

Therefore, we obtain the following corollary:

\begin{corollary}\label{corollary_1}
    Under \IBP valuations for the centers, the conditions of Theorem~\ref{thm:intra:item-wise} apply. This has, among others, the following consequences.
    \begin{itemize}[noitemsep,topsep=3pt]
        \item If the agents have additive valuations, then we can guarantee $\frac{\text{0.618-EFX}}{\text{intra 0.618-EFX}}$, implied by \cite{AMN20}.
        \item For agents with additive valuations that have a common ranking over the goods w.r.t. their value, we can guarantee $\frac{\text{EFX}}{\text{intra-EFX}}$, using \cite{plaut2020almost}. The same is true for bivalued valuations by \cite{amanatidis2021maximum}.
        \item For any fairness notion $F$ that is known to exist for monotone valuations without levels, we can guarantee  $\frac{\text{F}}{\text{intra-F}}$ when the agents have monotone valuations. 
    \end{itemize} 
\end{corollary}

\begin{proof}It is immediate that, under \IBP valuations, the value of a center depends solely on the subset of items it receives and not on how those items are allocated among its agents. Therefore, the assumptions of Theorem~\ref{thm:intra:item-wise} are satisfied. This property, combined with the monotonicity established in Lemma~\ref{lemma:intra:monotonicity}, ensures the third consequence. Finally, for the first two consequences, the claims follow by invoking Lemma~\ref{lemma:additivity-ibp}. 
\end{proof}
    
For bundle-based valuations, the premises of Theorem \ref{thm:intra:item-wise} do not apply and we leave it as an open problem for future work. We also note that making progress for this case can yield improvements on finding  balanced EF1 allocations in standard fair division 
(see Proposition \ref{prop:balanced} in Section~\ref{sec:conclusions}).
\section{Inter-fairness via Horizontal Round Robin and its variants}
\label{sec:inter}
\ifshort

Inter-fairness turns out to be much more challenging, as one needs to be very careful on how the goods within a center are assigned, so that agents from another center do not feel unhappy. We have made initial progress on special types of valuations, which still requires a non-trivial analysis. 
\fi

\subsection{Agents with ordered or identical valuations}
\label{subsec:ordered}
A natural and widely studied special case is when all agents share the same ranking of goods;
commonly referred to as \emph{ordered valuations}. An even more restricted special case arises when agents have \emph{identical valuations}. This has also received significant attention in the literature and serves as a useful benchmark. In our setting, we first focus on the variants in which agents share a common ranking or have identical valuations only \emph{within each center}. This assumption is motivated for scenarios where agents may be assigned to centers based on similar preferences and behavior.
Then, we investigate the case where all agents \emph{across} all centers have identical valuations. In these special cases, we can obtain fairness guarantees for the variants of center valuations discussed above via a simple adaptation of the well-known Round-Robin algorithm to our hierarchical setting. 

\medskip \noindent \textbf{\textsc{Horizontal Round-Robin (HRR).}} First, the algorithm fixes an arbitrary ordering of the centers and, within each center, an arbitrary ordering of its agents. Then, it allocates the items according to this fixed order. 
It is helpful to think of the agents as if they were arranged in a table with $n$ rows and $k$ columns. Each column corresponds to a center, and the $i$-th row consists of the agents $(a_{(i,1)}, a_{(i,2)}, \ldots, a_{(i,k)})$. During the execution of the algorithm, the items are allocated parsing the agents \emph{row by row}. While parsing row $i$, the algorithm visits the agents $a_{(i,1)}, a_{(i,2)}, \ldots, a_{(i,k)}$, allowing each agent to pick her favorite unallocated item. We refer to a complete parsing over a row as a \emph{small round}. A sequence of $n$ small rounds, corresponding to a complete parsing of
all rows, is called a \emph{big round}. Due to this row-by-row execution, we call this algorithm
\textsc{Horizontal Round-Robin}. The algorithm is presented as Algorithm~\ref{alg:horizontalRR}.

\begin{algorithm}[h]
\caption{\textsc{Horizontal Round-Robin}}
\label{alg:horizontalRR}
\DontPrintSemicolon

Fix an arbitrary ordering $C_1, C_2, \dots, C_k$ of the centers and, for each center $C_j$, an arbitrary ordering $a_{(1,j)}, a_{(2,j)}, \dots, a_{(n,j)}$ of its agents\;

\For{$i = 1$ \textbf{to} $n$}{
    \For{$j = 1$ \textbf{to} $k$}{
        $A_{(i,j)} \gets \emptyset$\;
    }
}

\While{$M \neq \emptyset$}{
    \For{$i = 1$ \textbf{to} $n$}{
        \For{$j = 1$ \textbf{to} $k$}{
            \If{$M = \emptyset$}{
                \textbf{break}\;
            }
            Let $g^\ast \in \arg\max_{g \in M} v_{a_{(i,j)}}(g)$\;
            $A_{{(i,j)}} \gets A_{(i,j)} \cup \{g^\ast\}$\;
            $M \gets M \setminus \{g^\ast\}$\;
        }
    }
}
\Return $\mathcal{B} = (B_1, \ldots, B_k)$, where $B_j = (A_{(1,j)}, A_{(2,j)}, \ldots, A_{(n,j)})$\;
\end{algorithm}

 Before turning to the proofs of our main theorems, we state the following simple yet useful observation. 

 \begin{observation}\label{lemma:HRR_interEF1}
    The HRR algorithm computes an inter-EF1 allocation for the
    agents in polynomial time, for additive valuation functions of the agents.
\end{observation}

The proof follows from the fact that \textsc{HRR} is equivalent to the standard Round-Robin algorithm in which agents are visited in the order $(a_{(1,1)}, \ldots, a_{(1,k)}, a_{(2,1)}, \ldots, a_{(2,k)}, \ldots, a_{(n,1)}, \ldots, a_{(n,k)})$.

\begin{theorem}[bundle-based potential valuations]\label{thm:common-rank:thm1}
When all agents within each center (but not necessarily across centers), agree on the ranking of the items, we can compute an $\frac{\text{EF1}}{\text{inter-EF1}}$ allocation 
in polynomial time, under either (i) \BBP or (ii) \BBR valuations.
\end{theorem}
\begin{proof}
    We prove the theorem separately for the two classes of valuations: (i) for \BBP valuations, and (ii) for \BBR valuations.
\begin{enumerate}
    \item[(i)] For \BBP valuations, we show that the \textsc{HRR} algorithm computes an $\frac{\text{EF1}}{\text{inter-EF1}}$ allocation. Due to Observation~\ref{lemma:HRR_interEF1}, it follows directly that the algorithm runs in polynomial time and the final allocation is inter-EF1 for the agents. Therefore, it remains to show that the resulting bundling profile $\mathcal{B}$ is EF1 for the centers.

    Let $\mathcal{B} = (B_1,\dots,B_k)$ be the computed bundling profile, where
    $B_j = \bigl(A_{(1,j)}, A_{(2,j)}, \dots, A_{(n,j)}\bigr)$ is the bundling assigned to center $C_j$.
    
    Below, for a given center $C_j$ we denote by $\succ_j$ the common ranking of agents in $C_j$.
    Thus, $g \succeq_j g'$ (resp., $g \succ_j g'$) means that \emph{every} agent in $C_j$ weakly (resp., strictly) prefers the good $g$ to $g'$.
    
    W.l.o.g., fix two centers $C_i$ and $C_j$ such that $C_i$ precedes $C_j$ in the ordering used by the algorithm. We first show that center $C_i$ does not envy $C_j$, i.e., $v_i(B_i) \geq v_i(B_j)$.
    By construction, in every small round of the algorithm, center $C_i$ selects a good before center $C_j$.
    Suppose the algorithm terminates after $\rho$ big rounds.
    Then either both centers have the same number of agents holding bundles of size $\rho$, or, since $C_i$ precedes $C_j$, center $C_i$ has exactly one more such agent. 
    
    Consequently, for any $x \in [n]$, $A_{(x,i)}$ and $A_{(x,j)}$ satisfy one of the following:
    \begin{itemize}
        \item both bundles contain $\rho-1$ items.
        \item both bundles contain $\rho$ items.
        \item $A_{(x,i)}$ contains $\rho$ items, while $A_{(x,j)}$ contains $\rho-1$ items.
    \end{itemize}
     In all cases, we have $|A_{(x,i)}| \geq |A_{(x,j)}|$. Let $\rho_x = |A_{(x,j)}|$.
    For $r \in [\rho_x]$, let $g_r^{(x,i)}$ and $g_r^{(x,j)}$ denote the items added in the $r$-th big round to bundles $A_{(x,i)}$ and $A_{(x,j)}$, respectively. Since $C_i$ precedes $C_j$ in the ordering, it holds that for all $\ell \in [n]$ and $x \in [n]$
    \begin{align} \label{eq:common-rank:eq1}
        g_{r}^{(x,i)} \succeq_i g_{r}^{(x,j)}.
    \end{align}
    By additivity of the agents' valuations and since $|A_{(x,i)}| \geq |A_{(x,j)}|,$ \eqref{eq:common-rank:eq1} implies that for all $\ell \in [n]$ and $x \in [n]$:
    \begin{equation} \label{eq:common-rank:eq2}
        v_{a_{(\ell,i)}}\bigl(A_{(x,i)}\bigr) \geq v_{a_{(\ell,i)}}\bigl(A_{(x,j)}\big).
    \end{equation}
    
    Now, let $\pi\in\Pi_n$ be a permutation of the bundles of $B_j$ that maximizes the social welfare of the agents in $C_i$. Consider the same permutation $\pi$, but now applied on the bundling $B_i$. Then, 
    \begin{align*}
    v_{C_i}(B_i) &= \max_{\pi' \in \Pi_n} \sum_{\ell=1}^{n} v_{a_{(\ell,i)}}\bigl(A_{(\pi'(\ell), i)}\bigr)\\ \nonumber
    &\geq \sum_{\ell=1}^{n}v_{a_{(\ell,i)}}\bigl(A_{(\pi(\ell),i)} \bigr)\\ \nonumber
    &\geq \sum_{\ell=1}^{n}v_{a_{(\ell,i)}}\bigl(A_{(\pi(\ell),j)} \bigr)\\
    &=v_{C_i}(B_j),
    \end{align*}
    where the first inequality holds because $\pi$ is a feasible (though not necessarily the optimal) permutation of the bundles of $B_i$ w.r.t. $C_i$, and the second follows from \eqref{eq:common-rank:eq2}.
    Thus, center $C_i$ does not envy center $C_j$.

    We proceed to show that center $C_j$ is EF1 w.r.t. $C_i$.
    Observe that, for each $x \in [n-1]$, a good is added to bundle $A_{(x,j)}$ before it is added to bundle $A_{(x+1,i)}$.
    Moreover, for any $x \in [n-1]$, bundles $A_{(x,j)}$ and $A_{(x+1,i)}$ satisfy one of the following: 
    \begin{itemize}
        \item both bundles contain $\rho-1$ items.
        \item both bundles contain $\rho$ items.
        \item $A_{(x,j)}$ contains $\rho$ items, while $A_{(x+1,i)}$ contains $\rho-1$ items. 
    \end{itemize}
    In all cases, we have $|A_{(x,j)}| \geq |A_{(x+1,i)}|$.
    Let $\rho_x = |A_{(x+1,i)}|$.
    Since agents in $C_j$ share a common ranking over the items, it follows that for all $r \in [\rho_x]$ and $x \in [n-1]$:
    \begin{equation} \label{eq:common-rank:eq3}
        g_{r}^{(x,j)} \succeq_i g_{r}^{(x+1,i)}.
    \end{equation}
    By additivity and since $|A_{(x,j)}| \geq |A_{(x+1,i)}|$, \eqref{eq:common-rank:eq3} implies that for all $\ell \in [n]$ and $x \in [n-1]$:
    \begin{equation} \label{eq:common-rank:eq4}
        v_{a_{(\ell,j)}}\bigl(A_{(x,j)}\bigr) \geq v_{a_{(\ell,j)}}\bigl(A_{(x+1,i)}\bigr).
    \end{equation}
    
    Finally, note that agent ${a_{(n,j)}}$ is the last agent in $C_j$ to be allocated an item.
    Thus, each good in $A_{(n,j)}$ is weakly preferred to the corresponding good in $A_{(1,i)}$ allocated in the next big round, except possibly the first item $g_1^{(1,i)}$ assigned to $a_{(1,i)}$.
    Therefore, for all $\ell \in [n]$ we have
    \begin{equation} \label{eq:common-rank:eq5}
        v_{a_{(\ell,j)}}\bigl(A_{(n,j)}\bigr)
        \geq v_{a_{(\ell,j)}}\bigl(A_{(1,i)}\bigr) - v_{a_{(\ell,j)}}\bigl(g_1^{(1,i)}\bigr).
    \end{equation}

    Now, let $\sigma \in \Pi_n$ be a permutation of the bundles of $B_i$ that maximizes the social welfare of the agents in $C_j$. Consider the same permutation $\sigma$ but now applied on $B_j$.
    Then,
    \begin{align*}
    v_{C_j}(B_j)
    &= \max_{\pi' \in \Pi_n} \sum_{\ell=1}^n v_{a_{(\ell,j)}}\!\bigl(A_{(\pi'(\ell),j)}\bigr) \nonumber \\
    &\geq \sum_{\ell=1}^n v_{a_{(\ell,j)}}\!\bigl(A_{(\sigma(\ell),j)}\bigr) \nonumber \\
    &\geq \sum_{\ell=1}^n v_{a_{(\ell,j)}}\!\bigl(A_{(\sigma(\ell),i)}\bigr)
       - v_{a_{(n,j)}}(g_1^{(1,i)}) \nonumber \\
    &= v_{C_j}\bigl(B_i \setminus g_1^{(1,i)}\bigr),
    \end{align*}
    where the first inequality holds because $\sigma$ is a feasible (though not necessarily the optimal) permutation of the bundles of $B_j$ w.r.t. $C_j$, and the second follows from (\ref{eq:common-rank:eq4}) and (\ref{eq:common-rank:eq5}). Thus, there exists a good whose removal makes $C_j$ non-envious of $C_i$. Hence center $C_j$ satisfies the EF1 condition w.r.t. center $C_i$, which completes the proof.

    \item[(ii)] 
    For \BBR valuations, we augment the \textsc{HRR} algorithm by adding the following step, just before the return statement. For each center \(C_j\), we compute a permutation of the bundles of \(B_j\) that maximizes the social welfare of the agents in \(C_j\), and assign them to the agents according to this permutation. That is, if \(\pi \in \Pi_n\) is such a permutation, then agent \(a_{(i,j)}\) receives the bundle \(A_{(\pi(i),j)}\), for all \(i \in [n]\) and all \(j \in [k]\). Such a permutation can be computed in polynomial time using a maximum matching algorithm. We prove that this modified version of the \textsc{HRR} algorithm computes an $\frac{\text{EF1}}{\text{inter-EF1}}$ allocation for \BBR valuations.

    First, observe that permuting the bundles among the agents within a center $C_j$ according to some permutation $\pi \in \Pi_n$ is equivalent to having fixed a different ordering of the agents within this center in the first step of HRR, namely the ordering induced by $\pi$. 
    Observe, such a permutation can be computed in polynomial time via a maximum weight matching in a bipartite graph where in one side we have the agents, in the other side the constructed bundles, and the weight of each edge is the value of the agent for the corresponding bundle.
    Since the order in which the agents of each center are visited during the execution of HRR is arbitrary, Observation~\ref{lemma:HRR_interEF1} continues to hold under this modification.
    
    Finally, if $\pi$ is the permutation of the bundles of the bundling $B_j$ that maximizes the social welfare of the agents in $C_j$, and the agents are assigned their bundles in the final step according to $\pi$, then the \BBR valuation coincides with the \BBP valuation. Therefore, by (i), the claim follows.
\end{enumerate}
\end{proof}

\begin{theorem}[item-based potential valuations]\label{thm:common-rank:thm2}
When all agents within each center (but not necessarily across centers) agree on the ranking of the items, HRR computes an $\frac{\text{EF1}}{\text{inter-EF1}}$ allocation 
in polynomial time, under \IBP valuations.
\end{theorem}

\begin{proof}
Similarly to the proof of Theorem~\ref{thm:common-rank:thm1}, by Observation~\ref{lemma:HRR_interEF1}, the algorithm runs in polynomial time and the final allocation is inter-EF1 for the agents. Therefore, it remains to show that the resulting bundling profile $\mathcal{B}$ is EF1 for the centers.

    Fix a center $C_i$. Since all agents in $C_i$ agree on the ranking of the goods with respect to their value, whenever an agent picks a good during a small round, that good is essentially of maximum value for the center among all remaining goods. 
    Although the selection is made by an individual agent, it can be viewed as the center itself choosing its most preferred good from the unallocated ones. Note that due to the algorithm, each center is visited exactly once during a small round. Moreover, since each center’s valuation depends only on the set of goods it receives, and not on how those goods are distributed among its agents, this process results in the same final allocation to the centers as a Round-Robin procedure in which the centers themselves act as agents and, when visited, pick their favorite available good. As a result, the final allocation is EF1 for the centers as well. This concludes the proof.
\end{proof}

\subsubsection{Identical Valuations} We now turn to the case of identical valuations within each center. Before presenting our results, we make the following useful observation regarding the center valuation variants in this setting:

\begin{observation}\label{obs:identical-all-coincide}
    Under identical valuations of the agents, for any given bundling, the \BBP, \IBP, \BBR, and \IBR valuations all coincide. This holds because any mapping of items and any permutation of the bundles of a given bundling to the agents of a center yield the same value for the center.
\end{observation}

Since this is a special case of ordered valuations, all results established above continue to apply. In addition to these results, the restriction imposed by identical valuations allows us to obtain $\frac{\text{EF1}}{\text{inter-EF1}}$ allocations under \IBR valuations. The following corollary summarizes these.

\begin{corollary}\label{corollary2}
    When all agents have identical valuations within each center (but not necessarily across centers), the HRR algorithm computes an $\frac{\text{EF1}}{\text{inter-EF1}}$ 
    in polynomial time, under any variant of center valuations discussed.
\end{corollary}
\begin{proof}
    Due to Observation~\ref{lemma:HRR_interEF1}, the algorithm runs in polynomial time and the final allocation is inter-EF1 for the agents. It remains to establish that the resulting bundling profile $\mathcal{B}$ is EF1 for the centers under each definition of center valuation.

    Under \BBP, \BBR and \IBP valuations, the claim follows immediately from Theorem~\ref{thm:common-rank:thm1} (for the first two) and Theorem~\ref{thm:common-rank:thm2}, respectively. Finally, due to Observation~\ref{obs:identical-all-coincide}, the claim holds for \IBR valuations as well.
\end{proof}

If we strengthen our assumptions to instances where all agents, \emph{across} all centers, share identical and monotone (not necessarily additive) valuations, we can achieve an even stronger guarantee. Specifically, we can establish the existence of an $\frac{\text{EF1}}{\text{inter-EFX}}$ allocation. To achieve this, the algorithm, which we refer to as \textsc{EFX-Partition Round-Robin}, proceeds as follows: 
\medskip

\noindent
\textsc{\textbf{EFX-Partition Round-Robin}}. First, we temporarily ignore the centers and treat the instance as a standard fair division problem with $n \cdot k$ identical agents. We run the algorithm by~\cite{plaut2020almost} (Algorithm 1) for identical agents to partition the goods into $n \cdot k$ bundles such that the resulting allocation is EFX. Then, the centers take turns picking bundles from this set in a standard round-robin fashion, picking the available bundle that gives the highest value, until all bundles are claimed. Finally, each center arbitrarily assigns the $n$ bundles it collected to its agents, assigning each agent a bundle. The algorithm is presented as Algorithm~\ref{alg:efx_rr}.

\begin{algorithm}[h]
\caption{\textsc{EFX-Partition Round-Robin}}
\label{alg:efx_rr}
\DontPrintSemicolon

Run the algorithm by~\cite{plaut2020almost} on $M$ for $n \cdot k$ identical agents (with valuation $v$) to partition $M$ into a set $S$ of $n \cdot k$ EFX bundles\;

Set $B_j \leftarrow \emptyset$ for every $j \in [k]$\;

Fix an ordering $C_1, \dots, C_k$ of the centers\;

\For{$j = 1$ \textbf{to} $k$}{
    Let $b^\ast \in \arg\max_{b \in S} v(b)$\;
    $B_j \gets B_j \cup \{b^\ast\}$\;
    $S \gets S \setminus \{b^\ast\}$\;
    
}

\For{$j = 1$ \textbf{to} $k$}{
    Arbitrarily assign the $n$ bundles in $B_j$ one-to-one to the $n$ agents in center $C_j$, forming $A_{(1,j)}, A_{(2,j)}, \ldots, A_{(n,j)}$\;
}

\Return $\mathcal{B} = (B_1, \ldots, B_k)$, where $B_j = (A_{(1,j)}, A_{(2,j)}, \ldots, A_{(n,j)})$\;
\end{algorithm}

\begin{theorem}\label{thm:identical-all-efx}
When all agents across all centers have identical and monotone (not necessarily additive) valuations, the \textsc{EFX-Partition Round-Robin} algorithm computes an $\frac{\text{EF1}}{\text{inter-EFX}}$ allocation, under any variant of center valuations discussed.
\end{theorem}
\begin{proof}
    First, we establish that the final allocation is inter-EFX for the agents. Since all agents (across all centers) have identical valuations, and the initial set of $n \cdot k$ bundles is generated to be EFX w.r.t. each other under the shared valuation function, the inter-EFX guarantee holds.

Next, we evaluate the fairness among the centers. Let $v(\cdot)$ be the shared valuation function of the agents. Consider two centers $C_i$ and $C_j$, the $i$-th and $j$-th in the Round-Robin order respectively, and assume w.l.o.g. that $i < j$. Order the $n \cdot k$ bundles in the set $S$ in descending order of their value under $v(\cdot)$ denoting them as $b_1, b_2, \dots, b_{n \cdot k}$. Note that some bundles in this set might be empty if the total number of goods is less than $n \cdot k$. Then, $C_i$ receives the bundles $b_i, b_{i+k}, \dots, b_{i+(n-1)k}$, while $C_j$ receives the bundles $b_j, b_{j+k}, \dots, b_{j+(n-1)k}$.

We will show that $C_i$ does not envy $j$. Specifically, since $C_i$ precedes $C_j$ in every round, for all $x \in \{0, 1, \dots, n-1\}$, it holds that $v(b_{i+xk}) \ge v(b_{j+xk})$. Therefore, we have:
$$v_{C_i}(B_i) = \sum_{x=0}^{n-1} v(b_{i+xk}) \ge \sum_{x=0}^{n-1} v(b_{j+xk}) = v_{C_i}(B_j).$$
Thus,  center $i$ is not envious of $j$.

Finally, we will show that $C_j$ is EF1 w.r.t. $C_i$. Because $C_j$ picks its $x$-th bundle before $C_i$ picks its $(x+1)$-th bundle, we have $v(b_{j+xk}) \ge v(b_{i+(x+1)k})$ for all $x \in \{0, \dots, n-2\}$. Furthermore, because all bundles are mutually EFX, we have that $v(b_{j+(n-1)k}) \ge v(b_i \setminus \{g^\ast\})$, where $g^\ast \in \arg\min_{g \in b_i} v(g)$. Summing these inequalities, we obtain:
$$v_{C_j}(B_j) = \sum_{x=0}^{n-1} v(b_{j+xk}) \ge v(b_i \setminus \{g^\ast\}) + \sum_{x=0}^{n-2} v(b_{i+(x+1)k}) = v_{C_j}(B_i \setminus \{g^\ast\}).$$
This completes the proof.
\end{proof}

While the preceding result establishes the existence of an $\frac{\text{EF1}}{\text{inter-EFX}}$ allocation, one might naturally wonder whether these agent-level and center-level guarantees can be inverted. As we demonstrate next, this is not the case. In fact, we show that an $\frac{\text{EFX}}{\text{inter-EF1}}$ allocation is not guaranteed to exist, even in the very restricted setting where agents are identical \emph{across} all centers.

\begin{theorem}\label{thm:impossibility-identical-efx}
For any variant of center valuations discussed, an $\frac{\text{EFX}}{\text{inter-EF1}}$ allocation may not always exist, even when all agents across all centers have identical  valuations. 
\end{theorem}

\begin{proof}
    Consider two centers $C_1$ and $C_2$, comprising agents $\{a_{(1,1)}, a_{(2,1)}\}$ and $\{a_{(1,2)}, a_{(2,2)}\}$, respectively, and a set of four items $M = \{g_1, g_2, g_3, g_4\}$. Assume that all agents have identical valuations over the items, and let $v_a(\cdot)$ denote their common valuation function. Let $v_a(g_1) = H$ and $v_a(g) = \epsilon$ for all $g \in M \setminus \{g_1\}$, where $H$ is a sufficiently large constant s.t. $H > 2\epsilon$ (see below).

\begin{center}
\small
\begin{tabular}{l|cccc}
\toprule
 & $g_1$ & $g_2$ & $g_3$ & $g_4$ \\
\midrule
$v_a(\cdot)$ & $H$ & $\epsilon$ & $\epsilon$ & $\epsilon$ \\
\bottomrule
\end{tabular}

\end{center}
Let $\mathcal{B} = (B_1, B_2)$ denote a bundling profile, where
$B_1 = (A_{(1,1)}, A_{(2,1)})$ and
$B_2 = (A_{(1,2)}, A_{(2,2)})$.
Let $G_1 = \bigcup_{\ell\in [n]} A_{(\ell,1)}$ and $G_2 = \bigcup_{\ell\in [n]} A_{(\ell,2)}$ denote the sets of items assigned to centers $C_1$ and $C_2$, respectively, under this bundling profile.
Due to Observation~\ref{obs:identical-all-coincide}, the arguments that follow apply to all variants among \BBP, \IBP, \BBR, and \IBR valuations.

Observe that for the EFX condition to hold for the centers, any allocation that assigns $g_1$ to a center must not assign any other items to that same center. To see this, suppose that $g_1 \in G_1$ and that $C_1$ receives at least one item from $\{g_2, g_3, g_4\}$. Then center $C_2$ receives at most two items from $\{g_2, g_3, g_4\}$, hence
$v_{C_2}(B_2) \leq 2\epsilon$.
However, we have that
$v_{C_2}(B_1) \geq H + \epsilon$.
Even after removing any small good from $G_1$, center $C_2$ continues to envy $C_1$, violating the EFX condition.

On the other hand, suppose that only $g_1$ is allocated to center $C_1$.
In this case, inter-EF1 is violated at the agent level. In $C_1$, there is one agent that receives nothing, who inevitably envies the agents in $C_2$. In particular, one agent in $C_2$ receives two items, and even after removing any single good from that agent’s bundle, the agent in $C_1$ with the empty bundle remains envious. 
\end{proof}

The above result comes in stark contrast to the standard setting, where EFX allocations always do exist for agents with identical valuations~\cite{plaut2020almost}. This result highlights the inherent difficulty of achieving stronger fairness guarantees in our setting. To complement this impossibility result, we conclude this section by showing that deciding whether such an allocation exists is computationally intractable, even in this restricted setting.

\begin{theorem}\label{hardness_theorem}
Deciding whether an instance admits an $\frac{EFX}{\text{inter-}EF1}$ allocation is NP-hard, even when all agents across all centers have identical and additive valuations.
\end{theorem}
\begin{proof}
    We reduce from the \textsc{Equal-Size Partition} (ESP) problem, which is known to be NP-hard~\cite{garey2002computers}. In an ESP instance, we are given a set of integers $S = \{s_1, s_2, \dots, s_{2m}\}$ such that $\sum_{i=1}^{2m} s_i = 2V$. The goal is to decide whether there exists a partition of $S$ into two subsets $X, Y \subset S$ such that $|X| = |Y| = m$ and $\sum_{s_i \in X} s_i = \sum_{s_i \in Y} s_i = V$.

Given an instance of ESP, we construct an instance of our problem as follows. Consider two centers $C_1$ and $C_2$, comprising agents $\{a_{(1,1)}, \dots, a_{(2m+2, 1)}\}$ and $\{a_{(1,2)}, \dots, a_{(2m+2, 2)}\}$, respectively. Thus, there are $4m+4$ agents in total. Let all agents across both centers have identical, additive valuations, denoted by their common valuation function $v_a(\cdot)$.

Let the set of goods $M$ consist of exactly $4m+4$ items, categorized as follows:
\begin{itemize}
    \item 2 \emph{huge} items $H_1, H_2$, with $v_a(H_1) = v_a(H_2) = 100mV$.
    \item $2m$ \emph{partition} items $g_1, \dots, g_{2m}$ corresponding to the elements of $S$, with $v_a(g_i) = 10m \cdot s_i$.
    \item $2m+2$ \emph{tiny} items $t_1, \dots, t_{2m+2}$, with $v_a(t_r) = 1$ for all $r \in [2m+2]$.
\end{itemize}

Let $\mathcal{B} = (B_1, B_2)$ denote a bundling profile, where $B_j = (A_{(1,j)}, \dots, A_{(n,j)})$ is the bundling assigned to center $C_j$. Due to Observation~\ref{obs:identical-all-coincide}, under identical valuations, any center's valuation simplifies to the total sum of the items in its bundling. We denote this common center valuation function as $v_C(\cdot)$, where $v_C(B_j) = \sum_{l=1}^{n} v_a(A_{(l,j)})$.

First, suppose that the ESP instance is a YES-instance. Then there exist disjoint subsets $X$ and $Y$ of size $m$ with equal sums $V$. We construct $\mathcal{B}$ by allocating $H_1$, the $m$ partition items corresponding to $X$, and $m+1$ tiny items to $C_1$. We allocate $H_2$, the $m$ partition items corresponding to $Y$, and the remaining $m+1$ tiny items to $C_2$. Within each center, these $2m+2$ items are assigned one-to-one to the $2m+2$ agents, such that $|A_{(l,j)}| = 1$ for all $l\in[2m+2]$ and $j\in[2]$. Because every agent receives exactly one item, inter-EF1 is trivially satisfied for the agents. For the centers, we have $v_C(B_1) = 100mV + 10mV + m + 1 = v_C(B_2)$, and thus the EFX condition is satisfied as well.

Conversely, suppose that the ESP instance is a NO-instance. We will show that the respective $\frac{EFX}{\text{inter-}EF1}$ instance is a NO-instance as well.

For the allocation to be inter-EF1, observe that every agent must receive exactly one item. Otherwise, since there are $4m+4$ total items and $4m+4$ agents, by the pigeonhole principle another agent must receive at least two items, violating inter-EF1. Consequently, each center must be allocated exactly $2m+2$ items.

Next, observe that for the allocation to satisfy the EFX condition at the center level, $H_1$ and $H_2$ cannot be allocated to the same center. If both are assigned e.g. to $C_1$, then $v_C(B_1) \ge 200mV$. In this case, center $C_2$ receives at most all of the remaining items, having $v_C(B_2) \le 20mV + 2m + 2$. Even after removing one huge item from $B_1$, $v_C(B_1 \setminus \{H_1\}) \ge 100mV > 20mV + 2m + 2 \ge v_C(B_2)$, which violates the EFX condition for $C_2$. Thus, w.l.o.g., $C_1$ receives $H_1$ and $C_2$ receives $H_2$.

Since each center receives exactly $2m+2$ items, including exactly one huge item, it must receive $2m+1$ non-huge items. Because there are only $2m$ partition items available, every center must receive at least one tiny item.

Let $X'$ denote the set of partition items allocated to $C_1$, and let $k = |X'|$. Center $C_1$ therefore receives $(2m+1-k)$ tiny items. Consequently, $C_2$ receives the remaining partition items $Y'$ (where $|Y'| = 2m-k$), alongside $(k+1)$ tiny items. The valuations of the centers are:
$$v_C(B_1) = 100mV + 10m \left(\sum_{g_i \in X'}s_i\right) + (2m+1-k)$$
$$v_C(B_2) = 100mV + 10m \left(\sum_{g_i \in Y'}s_i\right) + (k+1).$$

Since we were given a NO-instance of ESP, this implies that either $|X'| \neq |Y'|$ or $\sum_{g_i \in X'}s_i\neq\sum_{g_i \in Y'}s_i$. We distinguish between these cases:

\textit{Case 1: Equal values ($\sum_{g_i \in X'}s_i= \sum_{g_i \in Y'}s_i$) but unequal sizes ($k \neq m$).}
In this case, the difference in center valuations relies entirely on the tiny items:
$$|v_C(B_1) - v_C(B_2)| = |(2m+1-k) - (k+1)| = |2m - 2k|=|2(m-k)|.$$
Because $k \neq m$ and $k$ is an integer, $|2(m-k)| \ge 2$. W.l.o.g., assume $v_C(B_2) \ge v_C(B_1) + 2$, meaning $C_1$ envies $C_2$. Since $C_2$ possesses at least one tiny item, EFX requires that the envy is eliminated upon its removal. However, hypothetically removing a tiny item $t$ reduces the bundle's value by exactly $1$, leaving $v_C(B_2 \setminus \{t\}) \ge v_C(B_1) + 1 > v_C(B_1)$, and thus center $C_1$ continues to envy $C_2$, violating EFX.

\textit{Case 2: Unequal values ($\sum_{g_i \in X'}s_i \neq \sum_{g_i \in Y'}s_i$).}
Since the original set $S$ contains integers, $\sum_{g_i \in X'}s_i$ and $\sum_{g_i \in Y'}s_i$ must differ by at least $1$. W.l.o.g., assume $$\sum_{g_i \in Y'}s_i \geq \sum_{g_i \in X'}s_i+1.$$ 
The valuation difference is:
$$v_C(B_2) - v_C(B_1) = 10m\left(\sum_{g_i \in Y'}s_i-\sum_{g_i \in X'}s_i\right) + (k+1) - (2m+1-k) \ge 10m + 2k - 2m.$$
Since $C_1$ can receive at most all $2m$ partition items, $k \le 2m$, which implies $2k - 2m \ge -2m$. Substituting this yields:
$$v_C(B_2) - v_C(B_1) \ge 10m - 2m = 8m.$$
Since $m \ge 1$, we have an envy difference of at least $8$. Removing a tiny item from $C_2$ (which it must possess) reduces its value by $1$. However, $8m - 1 > 0$, so $C_1$ would still envy $C_2$, violating EFX.

In all cases where the ESP instance is a NO-instance, the respective $\frac{EFX}{\text{inter-}EF1}$ instance is a NO-instance as well. This completes the reduction and establishes the NP-hardness of our problem.
\end{proof}

\subsection{Agents with bivalued preferences and centers with item-based valuations}
\label{subsec:binary}
In this part, we move on to another well-studied case in the fair division literature, namely when agents have 
bivalued valuations. 
There,
each agent assigns one of two possible values to every item, typically referred to as a \emph{high} and a \emph{low} value.  
We focus on the case in which all agents within a center agree on the low and high values, but different centers may have different such values. We refer to this setting as \emph{center-specific bivalued valuations}. Accordingly, we denote by $l_i$ and $h_i$ the low and high value, respectively, for center $C_i$, for $i \in [k]$, where $l_i \leq h_i$.

We begin by studying the two variants of item-based valuations for the centers, before moving on to the more demanding case of bundle-based valuations, which we study in Section~\ref{sec:yankee}. 

One might hope that bivalued valuations for the agents could allow for stronger fairness guarantees. However, in the proof of Theorem~\ref{thm:impossibility-identical-efx}, the valuations of the agents are bivalued, hence, $\frac{\text{EFX}}{\text{inter-EF1}}$ is unattainable. 
Unfortunately, even the $\frac{\text{EF1}}{\text{inter-EF1}}$ guarantee cannot be obtained for \IBR valuations of the centers, even in the very simple setting of two centers, each comprising two agents. 
Intuitively, this is perhaps not surprising, since the way a center evaluates its own bundling against that of another center is unfavorable for itself: a center evaluates its own bundling, which was constructed to satisfy inter-EF1 among the agents, against the best possible reallocation of another center’s items, which aims to maximize value. This asymmetry poses an inherent difficulty.

\begin{theorem}\label{thm:impossibility_realized_EF1}
    Under \IBR valuations of the centers, an $\frac{\text{EF1}}{\text{inter-EF1}}$ allocation may not always exist, even when
    there are only two centers, each having two agents with bivalued valuations.
\end{theorem}
\begin{proof}
     Let $C_1$ and $C_2$ be two centers, comprising agents $\{a_{(1,1)}, a_{(2,1)}\}$ and $\{a_{(1,2)}, a_{(2,2)}\}$, respectively. Assume there are  $12$ goods, i.e., $M=\{g_1, g_2,\ldots, g_{12}\}$, and $v_{a_{(1,1)}}(g) = v_{a_{(2,1)}}(g) = v_{a_{(1,2)}}(g) = 1$ for all $g\in M$, while $v_{a_{(2,2)}}(g) = \epsilon$ for all $g\in M$, for some $0 < \epsilon \ll 1$. Due to this distinct valuation, we refer to agent $a_{(2,2)}$ as the \emph{special} agent. 

First, observe that the maximum value a center can obtain from any bundling equals the number of goods in that bundling, since each good is valued at 1 by at least one agent in the center. 

In order for an allocation to be EF1 at the center level, each center must receive exactly $6$ goods. If this is not the case and one were to receive at least $7$ goods, the other would receive at most $5$, and would necessarily envy even after the removal of one good.

Additionally, 
in any inter-EF1 allocation for the agents, no agent in $C_1$ can receive more than $3$ goods, while for $C_2$ each agent should receive $3$ goods. Let $B_1$ and $B_2$ denote the bundlings allocated to centers $C_1$ and $C_2$ respectively. Then, $v_{C_2}(B_2) =3 +3\epsilon$. On the other hand, $v_{C_2}(B_1) = 6$. 
Hence, the EF1 condition is violated for $C_2$ w.r.t. $C_1$, since for any $g\in B_1$, $v_{C_2}(B_1\setminus g) = 5 > 3+3\epsilon$.
\end{proof}

Due to this impossibility result, we now turn to the variant of \IBP valuations of the centers. Fortunately, a modification of HRR guarantees $\frac{\text{EF1}}{\text{inter-EF1}}$ for center-specific bivalued valuations. Compared to \textsc{HRR}, when an agent is visited and all available goods have low value for her, she picks a good that another agent values as high, if such a good exists, instead of picking a good arbitrarily. This good provides high value to the center in which she belongs, without harming her own, since she is in fact indifferent for the remaining goods, which all have low value for her. This algorithm is presented as Algorithm~\ref{alg:bivalued-item-based}.
\begin{algorithm}[h]
\caption{Center-oriented \textsc{Horizontal Round-Robin}}
\label{alg:bivalued-item-based}
\DontPrintSemicolon

Fix an ordering $C_1, \dots, C_k$ of the centers and, for each center $C_j$, an
ordering $a_{(1,j)}, \dots, a_{(n,j)}$ of its agents\;

For each center $C_j$, let $H_j = \{ g \in M \mid \text{there exists an agent } a_{(i,j)} \in N_j \text{ such that } v_{a_{(i,j)}}(g) = h_j \}$\;

Set $A_{(i,j)} \leftarrow \emptyset$ for every $i \in [n]$ and $j \in [k]$\;

\While{$M \neq \emptyset$}{
    \For{$i = 1$ \textbf{to} $n$}{
        \For{$j = 1$ \textbf{to} $k$}{
            \If{$M = \emptyset$}{
                \textbf{break}\;
            }
            
            \eIf{$H_j \cap \arg\max_{g \in M} v_{a_{(i,j)}}(g) \neq \emptyset$ \label{line:ifHj}}{
                Let $g^\ast \in \left(H_j \cap \arg\max_{g \in M} v_{a_{(i,j)}}(g)\right)$\;
                $A_{(i,j)} \gets A_{(i,j)} \cup \{g^\ast\}$\;
                $H_j \gets H_j \setminus \{g^\ast\}$\; \label{line:endif}
            }{
                Let $g^\ast \in \arg\max_{g \in M} v_{a_{(i,j)}}(g)$\;
                $A_{(i,j)} \gets A_{(i,j)} \cup \{g^\ast\}$\;
            }

            $M \gets M \setminus \{g^\ast\}$\;
        }
    }
}
\Return $\mathcal{B} = (B_1, \ldots, B_k)$, where $B_j = (A_{(1,j)}, A_{(2,j)}, \ldots, A_{(n,j)})$\;
\end{algorithm}
We obtain the following theorem.

\begin{theorem}\label{unlabelled theorem}
    Under center-specific bivalued valuations of the agents and \IBP valuations of the centers, Algorithm~\ref{alg:bivalued-item-based} computes an $\frac{\text{EF1}}{\text{inter-EF1}}$ allocation in polynomial time.
\end{theorem}
\begin{proof}
    It is clear that Algorithm~\ref{alg:bivalued-item-based} runs in polynomial time, similarly to Algorithm~\ref{alg:horizontalRR}. Moreover, observe that, as in Algorithm~\ref{alg:horizontalRR}, each agent belonging to center $C_j$ receives items she values at $h_j$ before receiving any good she values at $l_j$. Consequently, whenever an agent is visited, she picks her favorite unallocated good. By an argument similar to that of Observation~\ref{lemma:HRR_interEF1}, the resulting allocation is inter-EF1 for the agents.

    It remains to establish the EF1 guarantee for the centers. Fix a center $C_j$, and let $H_j$ be the set of goods that have high value for $C_j$, that is, $H_j = \{ g \in M \mid \text{there exists an agent } a_{(i,j)} \in N_j \text{ such that } v_{a_{(i,j)}}(g) = h_j \},$ where $N_j$ is the set of agents belonging to $C_j$. Let $L_j$ be the set of the remaining goods, which have low value for center $C_j$, i.e., $L_j = M \setminus H_j$.

    Observe that whenever $C_j$ receives a good, it is from $H_j$, except in the case where no goods from $H_j$ are available, in which case it starts receiving goods from $L_j$. To see this, let $a_{(i,j)}$ be an agent of center $C_j$ who is visited during some small round $\rho$. There are three possible cases regarding the good that agent $a_{(i,j)}$ picks:
    \begin{enumerate}
        \item[(i)] If there exists an available good $g^\ast$ that agent $a_{(i,j)}$ values as high, then she picks such a good. Since $g^\ast \in H_j$, this contributes value $h_j$ to the center.
    
        \item[(ii)] Otherwise, if all remaining unallocated goods have low value for agent $a_{(i,j)}$, but there exists a good that another agent of center $C_j$ values as high, then  it is still true that the intersection of $H_j$ with $\arg\max_{g \in M} v_{a_{(i,j)}}(g)$ is non-empty. Hence by lines \ref{line:ifHj}-\ref{line:endif} of the algorithm, agent $a_{(i,j)}$ picks such a good $g^\ast \in H_j$. Again, since $g^\ast \in H_j$, this contributes value $h_j$ to the center. This case is actually the main reason the algorithm works for the potential valuation of the center but not for the realized valuation.
    
        \item[(iii)] Otherwise, if all goods that are valued as high by some agent in $C_j$ are unavailable, that is, all goods in $H_j$ have already been allocated, then agent $a_{(i,j)}$ picks a good that she (and any other agent of $C_j$) values as low. In this case, since $g^\ast \in L_j$, this contributes value $l_j$ to the center.
    \end{enumerate}
    
    Since the valuation of each center depends only on the set of goods it receives, and not on how these goods are distributed among its agents, this process has the same outcome, in terms of the final allocation to the centers, as a Round-Robin algorithm in which the centers themselves act as agents and, when visited, select their most preferred available good. Therefore, the resulting allocation is EF1 for the centers as well. This concludes the proof.
\end{proof}
\section{Inter-fairness via Bilevel Yankee Swap: Agents with binary preferences and centers with bundle-based valuations}
\label{sec:yankee}
In this section we design a polynomial-time algorithm that always finds an $\frac{\text{EF1}}{\text{inter-EF1}}$ allocation when agents have binary valuations and centers have bundle based valuations.  We note that our algorithm produces a Pareto optimal allocation, i.e., it allocates items only to agents that value them 1, and thus it works for both realized and potential bundle-based valuations.

\subsection{Necessity for a new algorithm}
\label{sec:necessity}
Someone would rightfully wonder whether a suitable variation of HRR could work, especially since agents have ``simple'' valuations. Unfortunately, this is not the case as, under HRR, as soon as an item is given to an agent it remains there until the end. In principle, this could be potentially fixed if somehow each time we could choose the ``right'' agent and in addition, the agent would choose the ``right'' item. 
However this does not appear to be an easy task.

\begin{figure}
\vspace{-0.3cm}
\centering
\begin{tabular}{c|ccccc}
\toprule
\textbf{Agents} & $g_1$ & $g_2$ & $g_3$ & $g_4$ & $g_5$ \\
\midrule
$a_{(1,1)}$ & \circled{1} & 0 & 0 & 0 & 0 \\
$a_{(2,1)}$ & 0 & 0 & \circled{1} & 0 & 0 \\
$a_{(3,1)}$ & 0 & 1 & 0 & 1 & 0 \\
$a_{(4,1)}$ & 0 & 1 & 1 & 1 & 1 \\[4pt]
\midrule
$a_{(1,2)}$ & 0 & \circled{1} & 0 & 1 & 1 \\
$a_{(2,2)}$ & 1 & 0 & 0 & 0 & 0 \\
$a_{(3,2)}$ & 0 & 0 & 1 & 0 & 0 \\
$a_{(4,2)}$ & 0 & 0 & 1 & 0 & 0 \\
\bottomrule
\end{tabular}
\caption{An instance in which HRR fails to produce an 
$\frac{\text{EF1}}{\text{inter-EF1}}$ allocation under binary agent valuations and bundle-based center valuations. Circled entries indicate the first three allocations made by HRR.}
\label{fig:hrr-fails-binary}
\end{figure}

Firstly, we show that HRR cannot compute an $\frac{\text{EF1}}{\text{inter-EF1}}$ allocation when centers have bundle based valuations and agents have binary valuations. 
Consider the instance depicted 
in Figure~\ref{fig:hrr-fails-binary}, after HRR has assigned item $g_1$ to agent $a_{(1,1)}$, $g_2$ to agent $a_{(1,2)}$, and $g_3$ to agent $a_{(2,1)}$. For ease of exposition, we circle the items according to the agent to which they are allocated. 
Then, in the next step HRR would have to  
allocate $g_4$ to some agent in $C_2$, and then allocate $g_5$ to some agent in $C_1$. 
Note that both the realized and the \BBP\ value of center $C_2$ for the bundling $\{g_1, g_4\}$ equal $1$, whereas the \BBP\ value for the bundling $\{g_1, g_3, g_5\}$ equals $3$. 
Hence, $C_2$ envies $C_1$, and is not EF1. Additionally notice that even if we allocated $g_4$ to agent $a_{(1,2)}$ and thus increase the valuation of $C_2$ to 2 instead of 1, then the EF1 condition for agent $a_{(3,1)}$ is not satisfied.

A different approach would be to deploy an algorithm that exchanges bundles---like the envy-cycle elimination algorithm~\cite{LiptonMMS04}---or items---like the Yankee Swap algorithm~\cite{viswanathan2025general}.
Since the envy-cycle elimination algorithm does not change the bundles of the agents but only swaps them, it faces the same problems as HRR; this is in addition to new problems that may emerge due to bundle swapping.
On the other hand, Yankee Swap allows exchanges of goods between agents, thus it can ``forgive'' the ``wrong'' choices of assignments made earlier. 

At a high level, the Yankee Swap algorithm, in the classic setting, when there is simply a set of agents without any centers, works as follows. The algorithm goes over the agents at a Round Robbin fashion and when it is the turn of agent $i$ she either gets an unallocated item that she values 1, or she {\em steals} an item from a different agent, if there is a so called {\em transfer path}. 
This is a path that transfers already allocated items between agents and the last agent in the path gets an item that has not been allocated yet.
Crucially, every agent in a transfer path gets an item that values 1, hence we have a Pareto improvement, while we maintain EF1 between the agents.

Although the vanilla version of Yankee Swap, that uses an arbitrary ordering over the agents, does not immediately work in the presence of centers, someone could hope that following the ``special'' ordering of HRR would produce an $\frac{\text{EF1}}{\text{inter-EF1}}$ allocation. Unfortunately, our next example shows that this is not the case.

\begin{example}\label{exmp:yankee-swap}
We show that a naive modification of Yankee Swap, in which the agents and the centers are visited according to the ordering implied by HRR, cannot compute an $\frac{\mathrm{EF1}}{\text{inter-EF1}}$ allocation for bundle-based valuations of the centers and agents with binary preferences over the items.

Assume there are two centers $C_1$ and $C_2$, comprising agents $\{a_{(i,1)}\}_{i\in[3]}$ and $\{a_{(i,2)}\}_{i\in[3]}$, respectively. Figure~\ref{fig:yankee-swap} illustrates the instance and the execution of the algorithm. Figure~\ref{fig:yankee-swap-table} shows the agents’ preferences together with the partial allocation computed by the algorithm, while Figure~\ref{fig:yankee-swap-graph} depicts the graph constructed after the allocation of goods $g_1$, $g_2$, and $g_3$, described as follows. 
A node $[a_{(i,j)}, g_p]$
implies that agent $a_{(i,j)}$ holds good $g_p$, and a node $[a_{(i,j)}, \cdot]$ implies that $a_{(i,j)}$ does not hold any good. Finally, there are nodes of the form $g_p$ for $p\in[4]$ indicating unallocated items. An edge from a node $[a_{(i,j)}, g_p]$ or $[a_{(i,j)}, \cdot]$ to a node $[a_{(i',j')}, g_{p'}]$ means that $a_{(i,j)}$ values good $g_{p'}$. Finally, an edge from a node $[a_{(i,j)}, g_p]$ or $[a_{(i,j)}, \cdot]$ to a node $g_{p'}$ means that $a_{(i,j)}$ values the unallocated good $g_{p'}$. A transfer path is a path originating from a node including an agent and pointing to a node corresponding to an unallocated good.

\begin{figure}[ht]
\centering

\begin{subfigure}{0.48\textwidth}
\centering
\begin{tabular}{c|cccc}
\toprule
\textbf{Agent} & $g_1$ & $g_2$ & $g_3$ & $g_4$\\
\midrule
$a_{(1,1)}$ & \circled{1} & 0 & 0 & 0 \\
$a_{(2,1)}$ & 0 & 0 & \circled{1} & 0 \\
$a_{(3,1)}$ & 0 & 1 & 0 & 1 \\
[4pt]
\midrule
$a_{(1,2)}$ & 0 & \circled{1} & 0 & 1 \\
$a_{(2,2)}$ & 1 & 0 & 0 & 0 \\
$a_{(3,2)}$ & 0 & 0 & 1 & 0 \\
\bottomrule
\end{tabular}
\caption{The valuations of the agents of centers $C_1$ and $C_2$ and the partial allocation (circled) computed by the algorithm.}
\label{fig:yankee-swap-table}
\end{subfigure}
\hfill
\begin{subfigure}{0.48\textwidth}
\centering
\begin{tikzpicture}[every node/.style={font=\large}, >=stealth]

\node (C1) at (0,-0.2) {$C_1$};
\node (C2) at (4,-0.2) {$C_2$};

\node (a11) at (0,-1.1) {[$a_{(1,1)},g_1$]};
\node (a21left) at (0,-1.8) {[$a_{(2,1)},g_3$]};
\node (a31left) at (0,-2.5) {[$a_{(3,1)},\cdot$]};

\node (a12right) at (4,-1.1) {[$a_{(1,2)},g_2$]};
\node (a22) at (4,-1.8) {[$a_{(2,2)},\cdot$]};
\node (a32) at (4,-2.5) {[$a_{(3,2)},\cdot$]};

\node (g4) at (2,-3) {$g_4$};

\draw[->] (a22) -- (a11);
\draw[->] (a32) -- (a21left);
\draw[->] (a31left) -- (g4);
\draw[->] (a12right.west) -- (g4);

\end{tikzpicture}
\caption{The transfer graph produced after goods $g_1$, $g_2$, and $g_3$ have been assigned by the algorithm. For the sake of brevity, only the edges relevant to the discussion are depicted.}
\label{fig:yankee-swap-graph}
\end{subfigure}

\caption{The application of the Yankee Swap algorithm using the HRR ordering in Example~\ref{exmp:yankee-swap}.}
\label{fig:yankee-swap}
\end{figure}
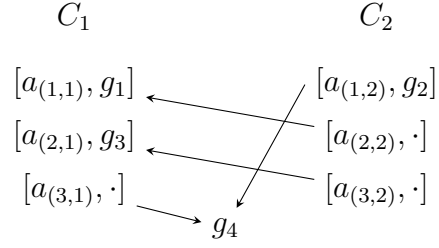

The algorithm begins with center $C_1$, where agent $a_{(1,1)}$ is assigned good $g_1$. It then moves to center $C_2$, where agent $a_{(1,2)}$ receives good $g_2$. Returning to center $C_1$, agent $a_{(2,1)}$ is assigned good $g_3$. At this point, it is the turn of center $C_2$, and there remains exactly one unallocated item, namely $g_4$. According to Figure~\ref{fig:yankee-swap-graph}, there is no transfer path from agent $a_{(3,2)}$. There is a transfer path from agent $a_{(1,2)}$, but assigning $g_4$ to her would violate the EF1 condition of agent $a_{(3,1)}$ towards her. Inevitably, the algorithm has to visit $C_1$. Then, there is a transfer path from agent $a_{(3,1)}$ towards good $g_4$. This allocation, however, violates the EF1 condition with respect to the centers. If $\mathcal{B} = (B_1, B_2)$ is the final allocation, observe that $v_{C_1}(B_1) = 1 < v_{C_1}(B_2 \setminus \{g\}) = 2$
for any $g \in \{g_1, g_3, g_4\}$.
\end{example}

The example above shows that sometimes it is inevitable to steal an item from an agent $i$ of center $C$, but give an item to a {\em different} agent of the same center. 
In order to bypass this difficulty, we need to define a new version of transfer paths, which we denote as {\em \squiz} paths (see Definitions~\ref{preferencegraph} and~\ref{def:squiz_paths}) and they are the main ingredient of our algorithm. 
At a high level, \squiz paths allow the exchange of items between agents {\em and} centers, and can ``take back'' one item from an agent $a_{(\ell,j)}$ and give a different item to a different agent $a_{(\ell',j)}$ of the same center.
This way, \squiz paths guarantee Pareto improvement at {\em center-level} and allow us to prove that our algorithm, which we call \algoname, produces an allocation that is fair at both levels. 

At a high level, \algoname\ works as follows.
The algorithm works in {\em epochs}, where in each epoch it follows an {\em epoch-specific} ordering over the centers, that is computed via the {\em center envy-graph}.
During each epoch, we maintain an {\em assignment} of items to agents, where each agent is assigned at most one item.
Crucially, the items in the assignment can change hands and centers many times.
Each epoch consists of {\em rounds}, where each round follows the ordering of the centers, until the centers have no agents that are {\em eligible} to be assigned an item in the current epoch.
Here, an agent is eligible if, at the present epoch, they have not been assigned an item and additionally there is a \squiz path that starts from them.
Hence, at the end of each epoch, each agent has either been assigned exactly one item that values 1, or there is no \squiz path that starts from them. 
At this point, the algorithm:
(a) augments the partial allocation of the previous epoch using the computed assignment of the present epoch, 
(b) updates the ordering of the centers, and 
(c) proceeds to the next epoch.
Crucially, at the end of each round the partial allocation is fixed, in the sense that no \squiz path can use any items that have been allocated.

\addtocounter{example}{-1}
\begin{example}[continued] 
We show that this suitably modified version of the Yankee Swap algorithm computes an $\frac{\mathrm{EF1}}{\text{inter-EF1}}$ allocation for this instance. Up to the allocation of $g_1, g_2, g_3$, \algoname\ proceeds as before, after which exactly one item remains unallocated, namely $g_4$. According to Figure~\ref{fig:yankee-swap-graph-mod}, there exists a transfer path originating from agent $a_{(1,2)}$, who is, however, not eligible. 
On the other hand, there is a \squiz path (see Figure~\ref{fig:yankee-swap-graph-mod}):
\[
[a_{(2,2)}, \cdot] \to [a_{(1,1)}, g_1] \to C_1 \to [a_{(3,1)}, \cdot] \to g_4 .
\]
When this path is resolved, results in agent $a_{(2,2)}$ ``stealing'' good $g_1$ from $a_{(1,1)}$, while center $C_1$ compensates for the loss of $g_1$ because agent $a_{(3,1)}$ receives good $g_4$. The resulting allocation $\mathcal{B} = (B_1, B_2)$ is depicted in Figure~\ref{fig:yankee-swap-table-mod}. This allocation is EF1 for the agents, since each agent holds at most one good, and EF1 for the centers, as $v_{C_1}(B_1) = v_{C_1}(B_2) = 2$ and $v_{C_2}(B_1) = v_{C_2}(B_2) = 2$.

\begin{figure}[ht]
    \centering
    \begin{subfigure}[b]{0.48\textwidth}
        \centering
        \begin{tikzpicture}[every node/.style={font=\large}, >=stealth]
            \node (C1) at (0,-0.2) {$C_1$};
            \node (C2) at (4,-0.2) {$C_2$};
        
            \node (a11) at (0,-1.1) {[$a_{(1,1)},g_1$]};
            \node (a21left) at (0,-1.8) {[$a_{(2,1)},g_3$]};
            \node (a31left) at (0,-2.5) {[$a_{(3,1)},\cdot$]};
            \node (a12right) at (4,-1.1) {[$a_{(1,2)},g_2$]};
            \node (a22) at (4,-1.8) {[$a_{(2,2)},\cdot$]};
            \node (a32) at (4,-2.5) {[$a_{(3,2)},\cdot$]};
            \node (g4) at (2,-3) {$g_4$};

            \draw[->, dotted, thick] (a22) -- (a11);
            \draw[->] (a32) -- (a21left);
            \draw[->, dotted, thick] (a31left) -- (g4);
            \draw[->] (a12right.west) -- (g4);

            \draw[->, bend right=90, dotted, thick] (C1.west) to (a31left.west);
            \draw[->, dotted, thick] (a11) -- (C1);            
            \end{tikzpicture}
            \caption{The \squiz graph produced after goods $g_1$, $g_2$, and $g_3$ have been assigned by the algorithm. For the sake of brevity, only the edges relevant to the discussion are depicted. 
            The chosen \squiz path is
           shown with dotted edges.}
        \label{fig:yankee-swap-graph-mod}
    \end{subfigure}
    \hfill
    \begin{subfigure}[b]{0.48\textwidth}
        \centering
        \begin{tabular}{c|cccc}
            \toprule
            \textbf{Agent} & $g_1$ & $g_2$ & $g_3$ & $g_4$\\
            \midrule
            $a_{(1,1)}$ & 1 & 0 & 0 & 0 \\
            $a_{(2,1)}$ & 0 & 0 & \circled{1} & 0 \\
            $a_{(3,1)}$ & 0 & 1 & 0 & \circled{1} \\[4pt]
            \midrule
            $a_{(1,2)}$ & 0 & \circled{1} & 0 & 1 \\
            $a_{(2,2)}$ & \circled{1} & 0 & 0 & 0 \\
            $a_{(3,2)}$ & 0 & 0 & 1 & 0 \\
            \bottomrule
        \end{tabular}
        \vspace{1.5em}
        \caption{The valuations of the agents and the final allocation (circled) computed by the algorithm.}
        \label{fig:yankee-swap-table-mod}
    \end{subfigure}
    \caption{The application of the \algoname\ algorithm in Example~\ref{exmp:yankee-swap}.}
    \label{fig:yankee-swap-modified}
\end{figure}
\end{example}

\subsection{Assignments, \squiz paths, and the center envy-graph}
\label{sec:algo-defs}
As we have already mentioned, \squiz paths are the main tool for our algorithm. 
In order to formally define \squiz paths, we first need to define assignments. 
In what follows, w.l.o.g. we will assume that for every item there is at least one agent that values it 1; otherwise we can simply give items that no one values to any agent.
\begin{definition}[Assignment]\label{assignmnet} 
    An assignment $\gamma$ is a partial allocation of items to the agents such that each agent: a) is given at most one item; b) if they are given an item, they value it 1.
\end{definition}

Observe the distinction between allocations and assignment. We will use the term allocation to refer to the items that are {\em fixed} to the agents and the term assignment to refer to the items that can change hands by our algorithm. 
We say that we {\em augment} an allocation $\alloc$ by an assignment $\gamma$, denoted $\alloc \augment \gamma$, if we add to the bundle of agent $i$ the item they are assigned under $\gamma$.

We use assignments to define the edges of the {\em \squiz graph}; each assignment $\gamma$ defines the \squiz (di)graph $G(\gamma)$. 

\begin{definition}[\squiz graph]
\label{preferencegraph}
Given an assignment $\gamma$, the \squiz graph $G(\gamma)$, is a directed graph defined as follows. There are three types of vertices: 
(a) for every center $C_j,$ we create a corresponding center-vertex;
(b) for every agent $a_{(\ell,j)}$, we create a corresponding agent-vertex;
(c) for every item $g$ that is not assigned under $\gamma$, we create a corresponding item-vertex.

There are four types of directed edges.

\begin{enumerate}
\item \label{type:one} \emph{Agent-vertex to Item-vertex}: for every agent $a_{(\ell,j)}$ and any (unassigned) item $g$ such that $v_{a_{(\ell,j)}}(g) = 1$, we add a directed edge between the corresponding vertices.

\item \label{type:two} \emph{Agent-vertex to Agent-vertex}: 
for every pair of agents $a_{(\ell,j)}$ and $a_{(\ell',j')}$, where under $\gamma$, agent $a_{(\ell',j')}$ is assigned an item $g$ such that $v_{a_{(\ell,j)}}(g) = 1$, we add a directed edge $a_{(\ell,j)}$ to $a_{(\ell',j')}$. 

\item \label{type:age_to_cen} \emph{Agent-vertex to Center-vertex}: for every agent $a_{(\ell,j)}$ that is assigned an item under $\gamma$ we add a directed edge from their corresponding vertex to the vertex representing center $C_j$.

\item \label{type:cen_to_age} \emph{Center-vertex to Agent-vertex}: 
for every center $C_j$ we add one edge from its 
vertex towards the 
vertex of every agent $a_{(\ell,j)}$ who belongs to this center and is not assigned any item under $\gamma$.
\end{enumerate}
\end{definition}

In essence, \squiz graph $G(\gamma)$ augments the ``standard'' transfer graph of Yankee Swap by adding directed edges of Type \eqref{type:age_to_cen} 
(agent-to-center) 
and of Type \eqref{type:cen_to_age} (center-to-agent). 
These edges though are crucial for defining \squiz paths and eligible agents.

\begin{definition}[\squiz path, eligible agents]
\label{def:squiz_paths}
Given an assignment $\gamma$, a \emph{\squiz} path in $G(\gamma)$ is a directed path that starts from an agent that has not been assigned any items under $\gamma$ and ends at an item-vertex. An agent is called \emph{eligible} if it is the starting vertex of a \squiz path in $G(\gamma)$.
\end{definition} 

Observe that \squiz paths essentially define the transfer of items between agents: every agent in a \squiz path gets the item they point to; agents whose vertices belong to Type \eqref{type:one} edges of a \squiz path get the unassigned item, while agents whose vertices belong to Type~\eqref{type:two} edges of a \squiz path ``steal'' the item of the agent they target. 
Additionally, note that \squiz paths can be easily computed using breadth first search for example.

What remains to define in order to be able to formally describe \algoname\ is how we compute the epoch-specific ordering of centers. 
We do this, via the {\em center envy-graph}.

\begin{definition}[Center envy-graph] 
    Given an allocation $\alloc$, the center envy-graph $H(\alloc)$ has one vertex for each center and there is a directed edge from vertex $i$ towards vertex $j$, if $C_i$ envies $C_j$. 
\end{definition}

Observe that when $H(\alloc)$ is acyclic, then there is an induced topological ordering of the centers. Our algorithm will guarantee the acyclicity of $H(\alloc)$, hence the epoch-specific ordering will be the topological order of the center envy-graph.

\subsection{\algoname }

We are now ready to formally describe and analyze the \algoname\ algorithm, which is presented as Algorithm~\ref{alg:main}.

\begin{algorithm}[h]
\caption{\textsc{\algoname}}
\label{alg:main}
\DontPrintSemicolon
{\bf Input:} Set of centers $C_1, \ldots, C_k$ each containing $n$ agents with binary valuations, and a set of items $M$\;
{\bf Output:} An $\frac{\text{EF1}}{\text{inter-EF1}}$ allocation.\;
Initialize $\alloc$ to be the empty allocation\;

\While{$M \neq \emptyset$}{
 Compute the ordering of the centers for the epoch via the center envy-graph $H(\alloc)$\; \label{step:ordering}
    Initialize $\gamma$ to be the empty assignment\;
    \While {there exist an eligible agent in the \squiz graph $G(\gamma)$}{ 
        Find the next center $C_j$ in the ordering that has an eligible agent $a_{(\ell, j)}$ in $G(\gamma)$\; \label{step:next_eligible_agent}
        Find an item $g \in M$ that is the end of a \squiz path $P$ starting from $a_{(\ell, j)}$\; \label{step:squiz_path}
        Update the assignment $\gamma$ using the \squiz path $P$\; \label{step:update_assignment}
        $M \leftarrow M \setminus \{g\}$\; \label{step:item_decrease}
    }
    Augment the allocation using the assignment of the epoch, i.e, $\alloc \leftarrow \alloc \augment \gamma$\; \label{step:update_allocation}
}

\end{algorithm}

\begin{theorem}\label{binary bundle based}
    Under binary valuations for agents, \BBR and \BBP valuations for centers, \algoname\ computes an $\frac{\text{EF1}}{\text{inter-EF1}}$ allocation in polynomial time.
\end{theorem}

Recall, at the end of each epoch any partial allocation $\alloc$ corresponds to a directed acyclic graph $H(\alloc)$. 
Although we will formally prove this claim later,  in what follows we will assume that in each epoch the ordering of the centers comes from a center envy-graph that is a DAG.

We first establish some structural properties of \squiz paths, for which we need the following notions.
An item $g$ is {\em transferred}, if it belongs to the \squiz path $P$ of Step~\ref{step:squiz_path}.
Also, given an assignment $\gamma$, an item $g$ is \emph{reachable} from an agent $a_{(\ell,j)}$ if there exists a directed path in the \squiz graph $G(\gamma)$ from $a_{(\ell,j)}$ either \textit{(i)}~to the vertex corresponding to item $g$ when $g$ is unassigned, or \textit{(ii)}~to the agent currently assigned item $g$ otherwise.
Similarly, an agent $a_{(\ell',j')}$ is said to be \emph{reachable} from $a_{(\ell,j)}$ under $\gamma$ if there exists a path in $G(\gamma)$ from $a_{(\ell,j)}$ to $a_{(\ell',j')}$.

\begin{lemma}\label{lemma1}
If an agent that has not been assigned an item under some assignment $\gamma$ is not eligible, then every item that is reachable from her in $G(\gamma)$, cannot be transferred and she cannot be part of any subsequently formed \squiz path.
\end{lemma}
\begin{proof}
    Suppose that, under the assignment $\gamma$, \algoname\ was searching for an eligible agent in a center $C_j$, and that at this point the agent $a_{(\ell,j)}$ became not eligible.  Let $T_{\gamma}$ the set of all \squiz paths in the \squiz graph $G(\gamma)$ and denote by $\Phi$ the set of reachable agents from $a_{(\ell,j)}$ and $\Psi$ the set of unreachable agents $a_{(\ell,j)}$ in $G(\gamma)$. Finally assume that $W_{\gamma}\subseteq M$ is the set of items that the algorithm has not yet allocated or assigned.

Since agent $a_{(\ell,j)}$ is not eligible then she cannot value any item in $W_{\gamma}$. Moreover no agent in $\Phi$ can lie on any \squiz path in $T_{\gamma}$, as otherwise $a_{(\ell,j)}$ could reach the part of the path that contains that agent which would mean that a \squiz path for $a_{(\ell,j)}$ would exist. This implies that the \squiz graph $G(\gamma)$ is partitioned into two subsets, $\Phi$ and $\Psi \cup W_{\gamma}$, such that no vertex in $\Phi$ has an outgoing edge to any vertex in $\Psi \cup W_{\gamma}$.

Since none of the agents that are reachable from $a_{(\ell,j)}$ value any unallocated or unassigned item in $W_{\gamma}$, nor any item assigned to an agent in $\Psi$, and since no \squiz path contains any agent in $\Phi$, it follows that in any subsequent assignment $\gamma'$ there can still be no edge from $\Phi$ to $\Psi \cup W_{\gamma'}$ as otherwise a \squiz path could already exist for $a_{(\ell,j)}$ when the algorithm was assigning items according to assignment $\gamma$. This in turn implies that agent $a_{(\ell,j)}$ can never be part of any subsequently formed \squiz path,  which in combination with the fact that the items are assigned only if they belong in a \squiz path, implies that none of the reachable items of $a_{(\ell,j)}$ can be \emph{transferred}.
\end{proof}

The structural property of \squiz paths that are revealed in Lemma~\ref{lemma1} forms the basis for proving Lemma~\ref{lemma3}.

\begin{lemma}\label{lemma3}
Consider an epoch $\rho$ and let $O^{(\rho)}$ be the ordering of the centers in this epoch. Let $C_i$ and $C_j$ be two centers with $C_i$ being parsed prior than $C_j$ according to $O^{(\rho)}$ and let $X^{(\rho)}_\ell$ for $\ell\in\{i,j\}$ denote the set of items that are allocated to center $C_\ell$ only in epoch $\rho$. Then the following are true: 
\ifshort
(1) $u_{C_i}\Bigl(X^{(\rho)}_i\Bigr) \geq u_{C_i}\Bigl(X^{(\rho)}_j\Bigr)$ and (2) $u_{C_j}\Bigl(X^{(\rho)}_j\Bigr) \geq u_{C_j}\Bigl(X^{(\rho)}_i\Bigr)-1$.
\fi
\end{lemma}
\begin{proof}
    Fix an epoch $\rho$ of the algorithm and let $O^{(\rho)}$, be the epoch-specific ordering of the centers and assume that $C_i$ and $C_j$ are two centers where $C_i$ is parsed prior to $C_j$ in $O^{(\rho)}$. Let $X^{(\rho)}_i$ and $X^{(\rho)}_j$ be the set of items that are allocated by the algorithm to the agents in centers $C_i$ and $C_j$ respectively only in epoch $\rho$. First we will argue that center $C_i$ cannot envy center $C_j$. The proof is obvious when $|X^{(\rho)}_i|\geq |X^{(\rho)}_j|$, since the algorithm allocates an item to some agent only if she values it as 1. If otherwise $|X^{(\rho)}_i| < |X^{(\rho)}_j|$, then we will show no agent in $C_i$, that did not receive any item in epoch $\rho$, can value any of the items in $|X^{(\rho)}_j|$.

Observe that since center $C_i$ is inspected prior to center $C_j$ and $|X^{(\rho)}_i|< |X^{(\rho)}_j|$, whenever the algorithm inspects an agent in $C_i$, for whom no \squiz  path exists, i.e, she is not eligible, then the following conditions hold:
\begin{enumerate}
    \item the algorithm has already allocated the same number of items to both centers; and
    \item there exists at least one agent in center $C_j$ who has not yet received an item and will be assigned one in a subsequent inspection of center $C_j$.
\end{enumerate}
Assume that $a_{(\ell,i)}\in C_i$ is an agent that did not receive an item in epoch $\rho$. W.l.o.g and from (1) we can assume when the algorithm at some time step $t$, inspected agent $a_{(\ell,i)}$, the items $P^{(t)} = \{p^1,\cdots ,p^t\}$ had been already assigned to some agents in $C_i$ and $Q^{(t)} = \{q^1,\ldots, q^t\}$ have been assigned to some agents in $C_j$. Then the following three statements are true.
\begin{itemize}
\item[(i)]Obviously, at that time step of the algorithm's execution, agent $a_{(\ell,i)}$
cannot value any unallocated item, as otherwise the algorithm could easily compute a \squiz path.
\item[(ii)] Agent $a_{(\ell,i)}$ cannot value any of the items in $Q^{(t)}$. Assume otherwise that $a_{(\ell,i)}$ values some item $q\in Q^{(t)}$ and let $a_{(\ell',j)}$ be the first agent in $C_j$ that satisfies condition (2). We know that $a_{(\ell',j)}$ will receive an item at some subsequent time step $t'>t$. Let $\text{P}_{a_{(\ell',j)}}$ denote the \squiz path that starts from agent $a_{(\ell',j)}$ and ends at some unallocated item $g$, at time step $t'$ and $P^{(t')}$ and $Q^{(t')}$ the set of items that $C_i$ and $C_j$ are assigned respectively at that time step. From Lemma \ref{lemma1}, none of the items in $Q^{(t')}$ that agent $a_{(\ell,i)}$ values can be part of $\text{P}_{a_{(\ell',j)}}$. Now since $a_{(\ell,i)}$ values item $q$, again from Lemma \ref{lemma1}, it holds that $g$ cannot be transferred and thus $q\in Q^{(t)}\cap Q^{(t')}$. But then notice that when the algorithm inspected $a_{(\ell,i)}$ it could have computed the following \squiz path: 
$$\text{P}': a_{(\ell,i)}\to q\to C_2 \to \underbrace{ a_{(\ell',j)}\to\ldots\to g}_{\text{P}_{a_{(\ell',j)}}}$$
where again due to Lemma \ref{lemma1}, the part of the path until the agent-vertex $a_{(\ell',j)}$ does not change in any future updates of the assignment within the epoch $\rho$. Consequently we get that $\text{P}'$ is a \squiz path for $a_{(\ell,i)}$ which contradicts the statement of Lemma \ref{lemma1}.

\item[(iii)] Agent $a_{(\ell,i)}$ cannot value any items that belong in the set $Q^{(z)}\setminus Q^{(t)}$, for any subsequent time step $z$ of the algorithm. This holds again from Lemma \ref{lemma1}, as none of the items that $a_{(\ell,i)}$ value can be transferred.
\end{itemize}

From the statements (i),(ii) and (iii) we conclude that no agent in $C_i$ that is not eligible can value any of the items in $|X^{(\rho)}_j|$ and thus $u_{C_i}\Bigl(X^{(\rho)}_i\Bigr)\geq u_{C_i}\Bigl(X^{(\rho)}_j\Bigr)$.

\smallskip We proceed to prove the same statement for center $C_j$. Again obviously since the algorithm allocates items to agents that value them as 1, if $|X^{(\rho)}_i|<|X^{(\rho)}_j|$ then center $C_j$ is not envious of center $C_i$.

It remains to show the case of $|X^{(\rho)}_i| > |X^{(\rho)}_j|$. First let us assume that $|X^{(\rho)}_i| > |X^{(\rho)}_j|+1$. The arguments that we will use are similar to the case for $C_i$. Since  the algorithm first parses $C_i$, then when some agent in $C_j$ is inspected and no \squiz path is found, the following two conditions hold:
\begin{enumerate}
\item the algorithm has allocated one more item to center $C_i$; and

\item there exists at least one agent in $C_i$ that has not yet received an item and will be assigned one in a subsequent inspection of $C_i$.
\end{enumerate}

Let $a_{(\ell,j)}\in C_j$, be an agent from $C_j$ that did not receive an item in epoch $\rho$. W.l.o.g assume that at the time step $t$, when the algorithm inspected $a_{(\ell,j)}$, the items $P^{(t)} = \{p^1,\ldots, p^t\}$ and $Q^{(t)} = \{q^1,\ldots, q^{(t-1)}\}$ had been already allocated to the agents in $C_i$ and $C_j$ respectively. Similarly to the above case notice the following statements are true.
\begin{itemize}
\item[(i)] Agent $a_{(\ell,j)}$ cannot value any unallocated item, as otherwise a \squiz path could be computed easily when $a_{(\ell,j)}$ was inspected.

\item[(ii)] Agent $a_{(\ell,j)}$ cannot value any of the items in $P^{(t)}$. Assume otherwise that $a_{(\ell,j)}$ values some item $p\in P^{(t)}$. Now let $a_{(\ell',i)}$ be the first agent in $C_i$ that satisfies condition (2) of the above observation. We know that $a_{(\ell',i)}$ will receive an item at some subsequent time step $t'>t$. Let $\text{P}_{a_{(\ell',i)}}$ denote the \squiz path that starts from $a_{(\ell,i)}$ and ends in some unallocated item $g$ at some subsequent time step $t'$. Again from Lemma \ref{lemma1}, none of the items in $P^{(t')}$ that agent $a_{(\ell,j)}$ values can be part of the path $\text{P}_{a_{(\ell',i)}}$. Since $a_{(\ell,j)}$ values $p$, again from Lemma \ref{lemma1}, $p$ cannot be transferred, thus $p\in P^{(t)}\cap P^{(t')}$. Now notice that when the algorithm inspected $a_{(\ell,j)}$ it could have computed the following \squiz path: 
$$\text{P}'': a_{(\ell,j)}\to p\to C_1 \to \underbrace{ a_{(\ell',i)}\to\ldots\to g}_{\text{P}_{a_{(\ell',i)}}}$$
where because the part of the path until the agent-vertex $a_{(\ell',i)}$ does not change in any future updates of the assignment within epoch $\rho$, we get that $\text{P}''$ is a \squiz path for $a_{(\ell,j)}$. This contradicts the statement of Lemma \ref{lemma1}.

\item[(iii)] Agent $a_{(\ell,i)}$ cannot value any item in the set $P^{(z)}\setminus P^{(t)}$ for any subsequent time step $z$ of the algorithm. This follows again from Lemma~\ref{lemma1}.
\end{itemize}
From statements (i),(ii) and (iii) we get that when $|X^{(\rho)}_i|> |X^{(\rho)}_j| +1$ then $u_{C_j}\Bigl(X^{(\rho)}_j\Bigr)\geq u_{C_j}\Bigl(X^{(\rho)}_i\Bigr)$.

Finally if  $|X^{(\rho)}_i| = |X^{(\rho)}_j| + 1$, from the fact that the algorithm allocates items only to agents that value them, it holds $u_{C_j}\Bigl(X^{(\rho)}_j\Bigr) \geq u_{C_j}\Bigl(X^{(\rho)}_i\Bigr) -1$, which states that $C_j$ might be envious of $C_i$ but satisfy the EF1 condition.
\end{proof}

Given two bundlings $B=(b_1, b_2, \ldots, b_n)$ and $Q = (q_1,q_2,\ldots,q_n)$, let
$B'=(b'_1, b'_2, \ldots, b'_n)$ be the \emph{extension of $B$ by $Q$ according to some permutation $\pi \in \Pi_n$}, where $b_i' = b_i \cup q_{\pi(i)}$ for all $i\in[n]$. 

\begin{lemma}\label{lemma4}
    Consider a center $C$, comprising agents $\{a_1, a_2, \ldots, a_n\}$, and two bundlings $B$ and $Q$.  Let $B'$ be the extension of $B$ by $Q$ according to some permutation $\pi \in \Pi_n$. Then, for the \BBP value of $B'$ it holds that $v_C(B') \leq v_C(B) + v_C(Q)$. 
\end{lemma}
The proof of the lemma is deferred to Appendix~\ref{app:sec:yankee}.

\begin{proof}[Proof of Theorem~\ref{binary bundle based}]
Throughout the proof we maintain the notation $X^{(\rho)}_\ell$ to denote the set of items that the algorithm allocates to agents in center $C_{\ell}$ only in epoch $\rho$ and denote by $B^{(\rho)}_{\ell}$ the bundling that the algorithm has allocated to the agents in center $C_{\ell}$ at the end of epoch $\rho$.

First observe that every time a \squiz path is computed the number of unallocated items decreases by exactly 1, while no item that has been allocated can return back to the pool of unallocated items. This fact guarantees that the algorithm terminates.

Also notice that since the algorithm allocates items only to agents that value them as 1, the \BBR and the \BBP valuations for the centers over the allocated bundlings are equal. Thus the same proof addresses both valuations of the centers. Additionally, for the same reason, for every center $C_{\ell}$ and every epoch $\rho$, it holds that $|X^{(\rho)}_\ell| = u_{C_{\ell}}\Big( X^{(\rho)}_{\ell} \Big)$. In the remainder of the proof, we use this fact implicitly.

To establish that \algoname\ computes an EF1 allocation for the centers, we use induction of the number of epochs that the algorithm needs to terminate. At the end of the first epoch, the EF1 condition is directly satisfied from Lemma \ref{lemma3}. Next consider some epoch $\rho$ and fix two centers $C_i$ and $C_j$. Assume w.l.o.g that in the next epoch $\rho$, $C_i$ is parsed prior to $C_j$ according to $O^{(\rho+1)}$. The inductive hypothesis at the end of epoch $\rho$ admits two cases, depending on the envy relationship between $C_i$ and $C_j$:
\begin{enumerate}
    \item neither center envies the other; 
    \item Since $C_i$ is prioritized, $C_i$ is EF1-envious of $C_j$ at the end of epoch $\rho$.
\end{enumerate}

We analyze each case separately below

\smallskip Case (1): In this case the inductive hypothesis is

\begin{equation}
u_{C_p}\Bigl( B^{(\rho)}_p\Bigr) \geq  u_{C_p}\Bigl( B^{(\rho)}_q\Bigr)\text{ for }p,q\in\{i,j\}\text{ and } p=\{i,j\}\setminus q
\end{equation}
Since $C_i$ is parsed prior to $C_j$ 
in epoch $\rho+1$, from Lemma \ref{lemma3} it holds that $u_{C_i}\Big( X^{(\rho+1)}_{i}\Big) \geq u_{C_i}\Big( X^{(\rho+1)}_{j}\Big)$. Consequently it follows that: 

\begin{equation}\label{eq1}
\begin{aligned}
u_{C_i}\Bigl(B^{(\rho+1)}_{i} \Bigr) &= u_{C_i}\Bigl(B^{(\rho)}_{i} \Bigr) + |X_i^{(\rho +1)}| = u_{C_i}\Bigl(B^{(\rho)}_{i} \Bigr) + u_{C_i}\Big(X^{(\rho+1)}_i\Big)
\end{aligned}
\end{equation}
where from the above argument and the inductive hypothesis it holds that
\begin{equation}
u_{C_i}\Bigl(B^{(\rho)}_{i} \Bigr) + u_{C_i}\Big(X^{(\rho+1)}_i\Big) \geq u_{C_i}\Bigl(B^{(\rho)}_{j} \Bigr) + u_{C_i}\Big(X^{(\rho+1)}_j\Big) 
\end{equation}
Observe that, in essence, the bundling $B^{(\rho+1)}_j$ is a superbundling (see Definition~\ref{superbundling}) of $B^{(\rho)}_j$,  resulting from the extension of $B^{(\rho)}_j$ by 
$\underbrace{(\{y_1\}, \{y_2\},\ldots, \{y_{|\textbf{Y}^{(\rho+1)}|}\}, \emptyset,\ldots,\emptyset)}_{n}$, where $y_i \in X_j^{(\rho+1)}$ for $i\in[|X_j^{(\rho+1)}|]$, according to some permutation $\pi \in \Pi_n$. Consequently from to Lemma \ref{lemma4} we get that $$ u_{C_i}\Bigl(B^{(\rho)}_{j} \Bigr) + u_{C_i}\Bigl(X_j^{(\rho +1)}\Bigr) \geq u_{C_i}\Bigl(B^{(\rho+1)}_{j} \Bigr)$$
which proves that center $C_i$ is EF1 towards $C_j$ at the end of epoch $\rho+1$. 

For center $C_j$ since it is parsed after $C_i$ in epoch $\rho+1$, from Lemma \ref{lemma3} it holds that $u_{C_j}\Big( X_j^{(\rho+1)}\Big) \geq u_{C_j}\Big( X_i^{(\rho+1)}\Big)-1$. Consequently using the same arguments as above we get that 
\begin{equation}
 u_{C_j}\Bigl(B^{(\rho+1)}_{j} \Bigr) \geq  u_{C_j}\Bigl(B^{(\rho)}_{j} \Bigr) + u_{C_j}\Big(X^{(\rho+1)}_j\Big) \geq u_{C_j}\Big( B^{(\rho)}_i\Big)+u_{C_j}\Big(X^{(\rho+1)}_i\Big)-1 \geq u_{C_j}\Big( B^{(\rho+1)}_i\Big)-1. 
\end{equation}

\smallskip Case (2): In this case the inductive hypothesis is as follows
\begin{equation}
u_{C_i}\Big(B^{(\rho)}_i\Big) \geq u_{C_i}\Big(B^{(\rho)}_j\Big) -1 \text{ and }u_{C_j}\Big(B^{(\rho)}_j\Big) \geq u_{C_j}\Big(B^{(\rho)}_i\Big)
\end{equation}
Similarly for $C_i$ it holds that $u_{C_i}\Big(X^{(\rho+1)}_{i}\Big) \geq u_{C_i}\Big(X^{(\rho+1)}_{j}\Big)$ and thus
\begin{equation}
u_{C_i}\Bigl(B^{(\rho+1)}_{i} \Bigr) = u_{C_i}\Bigl(B^{(\rho)}_{i} \Bigr)  + u_{C_i}\Big(X^{(\rho+1)}_i\Big) \geq u_{C_i}\Bigl(B^{(\rho)}_{j} \Bigr) -1 + u_{C_i}\Big(X^{(\rho+1)}_j\Big) \geq  u_{C_i}\Big(B^{(\rho+1)}_j\Big)-1
\end{equation}
Observe that for center $C_j$, since we assumed that in epoch $\rho$ it is envied by $C_i$, then it is a necessary condition that $|B^{(\rho)}_j| > |B^{(\rho)}_i|$, which consequently means that $u_{C_j}\Big( B^{(\rho)}_j\Big) \geq u_{C_j}\Big( B^{(\rho)}_i\Big) +1$, thus the following hold for $C_j$

\begin{equation}
u_{C_j}\Bigl(B^{(\rho+1)}_{j} \Bigr) = u_{C_j}\Bigl(B^{(\rho)}_{j} \Bigr)  + u_{C_j}\Big(X^{(\rho+1)}_j\Big) \geq u_{C_j}\Bigl(B^{(\rho)}_{i} \Bigr)+1 + u_{C_j}\Big(X^{(\rho+1)}_i\Big) -1 \geq  u_{C_i}\Big(B^{(\rho+1)}_i\Big)
\end{equation}
which establishes that if at the end of epoch $\rho$ the centers satisfy the EF1 condition then they do so also at the end of the next epoch. 

Notice also that these statements establish that no two centers can envy each other at the end of any epoch, which shows that the center envy-graph is acyclic at the end of each epoch.

Next we establish the inter-EF1 condition for the agents. Observe that an agent’s total valuation for the bundle allocated by the algorithm equals the number of epochs in which she received an item. Consider an agent $a_{(\ell,j)}\in C_j$, and let $\rho$ be the first epoch in which she does not receive an item. Then $a_{(\ell,j)}$ has valuation $\rho - 1$ for her bundle and since $\rho$ is the first epoch in which $a_{(\ell,j)}$ does not receive an item, she cannot envy any agent whose bundle size is smaller than $\rho$. However, she may envy agents whose bundles increased from size $\rho - 1$ to $\rho$, as they may have received in epoch $\rho$ an item that she values. Finally, because $a_{(\ell,j)}$ does not value any item that remains unallocated at the end of epoch $\rho$, her EF1 envy status toward any agent who later receives such items does not change. This proves that \algoname\ also computes an inter-EF1 for the agents.
\end{proof}
\section{Discussion}
\label{sec:conclusions}
 We studied the problem of fair division in a two-level setting where agents are partitioned into centers. We focused on satisfying fairness criteria simultaneously at both levels. 
Our results reveal that choices of valuation function for the centers and the information structure qualitatively change the results we can obtain and they demonstrate that delegated fair division is inherently different than the standard case. 
Furthermore, our results show that we require new algorithmic techniques and structural insights in order to derive some positive results.

We believe we have only scratched the surface of the area, since there is a plethora of future avenues to study. These include: 

\begin{enumerate}
    \item 
{\bf Different valuation function for the centers.} 
A key part of our contribution was defining the valuation of a center based on the preferences of its members.
 An alternative approach to our definitions would be to define a center's valuation for a bundling as the maximum social welfare achievable by a reallocation to its agents that is also fair (e.g., EF1) for them. We did not pursue this direction for two reasons. First, in the \emph{bundle-based} setting, an EF1 permutation of the bundlings may not exist. While one could assign a value of zero in that case, implying that a center prefers receiving nothing over a bundling that cannot be distributed fairly among its agents, this seemed somewhat unnatural, although easier to satisfy. Second, in the \emph{item-based} setting, finding a reallocation that maximizes social welfare subject to EF1 is computationally intractable~\cite{aziz2023computing}.

\item \label{open:one-to_one}
{\bf One-to-one allocations.}
One special case that we believe it deserves further study is when there are $k\cdot n$ items to allocate, i.e., each agent is ``entitled'' to exactly one item. This case is interesting for several different reasons. Firstly, it can serve as a stepping stone towards the more general case where agents have general additive valuations.
A different, not directly obvious, reason is that a positive result would immediately imply the existence of {\em balanced} EF1 allocations for a subclass of monotone valuations; in balanced allocations all agents get the same number of items. 
 While such allocations can be found via Round Robin algorithm when agents have additive utilities, for monotone functions the space is wide open. 
 \cite{kyropoulou2020almost} proved that for two agents with general monotone valuations there is always a balanced EF1 allocation.
 Recently,~\cite{kawase2026fair} designed polynomial-time algorithms that compute balanced allocations that are EF1 and fractional Pareto optimal when the agents have bivalued valuations, or when there are two types of agents.
 However, it is unknown whether balanced EF1 allocations always exist for more general settings.

To elaborate on the implications for balanced EF1 allocations, consider the problem \textsc{Balanced-EF1}, defined as follows: We are given $k$ agents and a set of $m=k\cdot n$ items for some $n$. The question is to compute an EF1 allocation where every agent receives exactly $n$ items (if such an allocation exists). We consider a special case of \textsc{Balanced-EF1} where each agent has a monotone, non-additive valuation. The valuation of each agent is defined as follows: agent $i$ is associated with a weighted, complete bipartite graph $G_i = ([n], [m], E)$, where the $kn$ vertices in the right part correspond to the items, and each edge has a positive weight. Given $S\subseteq [m]$, we let $v_i(S)$ to be equal to the maximum matching in the subgraph of $G_i$ induced by $S$ (i.e. in the complete bipartite graph between $[n]$ and $S$). We refer to these valuations as \emph{matching valuations}. The following claim is easy to see and places these valuations within the broader class of subadditive functions.

 \begin{claim}
     The valuation function of each agent is non-additive and falls within the class of fractionally subadditive, i.e, XOS valuation functions. 
 \end{claim}

Before we proceed, observe that, when $|M|=kn$ and every agent assigns strictly positive value to every item, every inter-EF1
allocation is one-to-one. Indeed, if some agent receives no item, then
some other agent receives at least two items; after the removal of any one
of these items, the former agent still envies the latter.
Conversely, every one-to-one allocation is trivially inter-EF1 (and hence
also intra-EF1). Thus, under these assumptions, the inter-EF1 (or intra-EF1) requirement is equivalent to requiring a one-to-one allocation.
 
\begin{proposition}
\label{prop:balanced}
\textsc{Balanced-EF1} for agents with matching valuations, reduces to the problem of computing an $\frac{EF1}{intra-EF1}$ allocation for equally sized centers with bundle based potential valuations and agents with positive value for every item.
\end{proposition}

\begin{proof}
Suppose we had an algorithm $A$ that can find an $\frac{EF1}{intra-EF1}$allocation  for the one-to-one setting, in particular for instances with $k$ centers and $kn$ items, with bundle-based potential valuations for the centers, and with all agents underneath each center having a positive value for all items. Let $I$ be an instance of \textsc{Balanced-EF1} where agents have \emph{matching valuations}. We construct an instance of our one-to-one setting $I'$ where each agent from $I$ corresponds to a center in $I'$. Furthermore, for each vertex $j\in [n]$ in the graph $G_i$, we add an agent $j$ underneath center $i$. The valuation function of $j$ is additive and is induced by the edges of $G_i$ that are adjacent to $j$. 

Suppose that algorithm $A$ returns an $\frac{EF1}{intra-EF1}$ allocation for the instance $I'$. Since it is intra-EF1 and all agents underneath the centers have positive values for all items, this means that every agent has received exactly one item. Therefore each center has received exactly $n$ items. Given also that the allocation is EF1 for the centers, under bundle-based potential valuations, this precisely means that the corresponding allocation in instance $I$ is balanced and EF1.
\end{proof}

In conclusion, if we can solve our one-to-one setting for bundle-based potential valuations, we can compute a balanced EF1 allocation for a special case of monotone, subadditive (but non-additive) valuations. 
 
 \item 
{\bf More levels of delegation.}
One can extend our model beyond two levels to capture scenarios like multi-level organizations; the hierarchical model of \cite{LBBM25} captures this type of scenarios, albeit it does not focus on fairness constraints. We note that, the results of Section~3 can be extended to any number of levels, provided that valuations at each level remain independent of the specific allocation of the goods at lower levels (as, e.g., in the case of item-based potential valuations). 

\item 
{\bf Centers of unequal sizes.}
Our model currently assumes centers are equally sized; that is, they have the same number of agents. A practical extension would be to consider centers with varying sizes. This could be addressed by ensuring that larger centers receive proportionally larger shares of resources by using weighted versions of the fairness notions. 
\end{enumerate}

Finally, while we focused on EF1 and EFX for goods, our model could be adapted to other fairness notions, such as MMS, as well as to different types of resources, including chores or mixed manna.

\subsection*{Acknowledgments}
Argyrios Deligkas and Stavros D. Ioannidis where supported by the  the EPSRC grant EP/X039862/1. 
Evangelos Markakis has been supported by the project MIS 5154714 of the National Recovery and Resilience Plan “Greece 2.0” funded by the European Union under the NextGenerationEU Program.

\bibliographystyle{plain}
\bibliography{bibl}

\appendix
\newpage

\clearpage
\section{Missing Proofs from Section \ref{sec:intra}}
\label{app:sec:intra}

\subsection{Proof of Lemma \ref{lemma:intra:monotonicity}}
  
Consider a center $C_j$. Let $B_j=(A_{(1,j)},A_{(2,j)},\dots,A_{(n,j)})$ be a bundling and $B_j'$ be a superbundling of $B_j$.
    
    Regarding \IBP valuations, observe that $\bigcup_{i\in[n]} A_{(i,j)} \subseteq \bigcup_{i\in[n]} A_{(i,j)}'$, and therefore, under the agents’ monotone valuations, $B'_j$ gives the center at least the same value as $B_j$. For \IBR valuations \footnote{Monotonicity is evaluated by comparing a center’s value for a bundling and a superbundling of it, assuming it were \emph{assigned} those bundles.}, the monotonicity property trivially holds since, under the agents’ monotone valuations,
    \begin{equation}\label{eq:intra:eq1}
    v_{C_j}(B_j) =\sum_{i\in[n]} v_{a_{(i,j)}}(A_{(i,j)}) \leq \sum_{i\in[n]} v_{a_{(i,j)}}(A'_{(i,j)}) = v_{C_j}(B'_j) .
    \end{equation}
    
    Similarly, for \BBP valuations, let $\pi \in \Pi_n$ be a permutation of the bundles in $B_j$ that maximizes the social welfare of the agents in $C_j$. Since the agents in $C_j$ have monotone valuations,
    \begin{equation*}
    v_{C_j}(B_j) = \sum_{i\in[n]} v_{a_{(i,j)}}\bigl(A_{(\pi(i),j)}\bigr) \leq \sum_{i\in[n]} v_{a_{(i,j)}}\bigl(A'_{(\pi(i), j)}\bigr) \leq v_{C_j}(B_j'),
    \end{equation*}
    where the last inequality holds since $\pi$ is a feasible permutation but not necessarily the optimal one for $B'_j$ w.r.t. center $C_j$. Finally, under \BBR valuations, the value of center $C_j$ coincides with the value obtained under \IBR valuations, and inequality~\eqref{eq:intra:eq1} applies again.

\subsection{Proof of Lemma            \ref{lemma:additivity-ibp}}
Consider a center $C_j$, and let $B_j=(A_{(1,j)},A_{(2,j)},\dots,A_{(n,j)})$ be a bundling allocated to it. Let $G = \bigcup_{i\in[n]} A_{(i,j)}$ denote the total set of items received by $C_j$. Observe that under item-based potential valuations and additive valuations of the agents, the valuation of the center can be defined over the items instead of the bundlings, as $v_{C_j}(g) = \max_{i \in [n]} v_{(i,j)}(g)$. Then, by the definition of \IBP valuations, we have that $v_{C_j}(B_j) = \sum_{g\in G} \max_{i \in [n]} v_{(i,j)}(g)$, and thus the valuation of the center is additive.

\section{Missing Proofs from Section \ref{sec:yankee}}
\label{app:sec:yankee}

\subsection{Proof of Lemma~\ref{lemma4}}

        Let $B=(b_1, b_2, \ldots, b_n)$, $Q = (q_1,q_2,\ldots,q_n)$ and  $B'=(b'_1, b'_2, \ldots, b'_n)$.
    We then get, successively, 
    \begin{align*}
        v_C(B') &= \max_{\pi'\in\Pi_n} \sum_{\ell=1}^n v_{a_\ell}(b'_{\pi'(\ell)})\\
        &=\max_{\pi'\in\Pi_n} \sum_{\ell=1}^n v_{a_\ell}(b_{\pi'(\ell)} \cup q_{\pi'(\pi(\ell))})\\
        &\leq \max_{\pi'\in\Pi_n} \sum_{\ell=1}^n v_{a_\ell}(b_{\pi'(\ell)}) + \max_{\pi'\in\Pi_n} \sum_{\ell=1}^n v_{a_\ell}(q_{\pi'(\pi(\ell))})\\
        &= \max_{\pi'\in\Pi_n} \sum_{\ell=1}^n v_{a_\ell}(b_{\pi'(\ell)}) + \max_{\pi'\in\Pi_n} \sum_{\ell=1}^n v_{a_\ell}(q_{\pi'(\ell)})\\
        &= v_C(B) + v_C(Q),
    \end{align*}
    and the lemma follows.

\newpage

\end{document}